\documentclass[10pt,twocolumn,twoside]{IEEEtran}
\usepackage[numbers]{natbib}
\ifCLASSINFOpdf
\else
\fi
\usepackage{graphics}
\usepackage{graphicx}
\usepackage{subcaption}
\usepackage{multirow}

\usepackage{amsmath,,amssymb}
\usepackage{amsthm}

\usepackage{pstricks, pst-node}
\usepackage{tikz}
\usepackage{pgfplots}
\usepackage{algorithmic}
\usepackage{algorithm}

\usepackage{url}

\newtheorem{theorem}{Theorem}
\newtheorem{lemma}{Lemma}

\def\II{{\mathbb I}}
\def\sign{{\rm sign}}
\def\Kcal{\mathcal K}
\def\RR{\mathbb{R}}
\def\PP{\mathbb{P}}
\def\Rcal{\mathcal R}
\def\EE{\mathbb{E}}
\def\Ncal{\mathcal N}
\def\Scal{\mathcal S}
\def\NN{\mathbb N}
\def\Tcal{\mathcal T}
\def\Ical{\mathcal I}

\begin{document}
\title{Automated and Robust Quantification of Colocalization in Dual-Color Fluorescence Microscopy: A Nonparametric Statistical Approach}

\author{Shulei~Wang$^{\ast,\dag}$,
        Ellen~T.~Arena$^{\dag,\ddag}$,
        Kevin~W.~Eliceiri$^{\dag,\ddag}$,
        and~Ming~Yuan$^{\ast,\dag,\ddag,\natural}$ % <-this % stops a space
\thanks{$^\ast$Columbia University}
\thanks{$^\dag$University of Wisconsin-Madison}
\thanks{$^\ddag$Morgridge Institute for Research}
\thanks{$^\natural$Address for Correspondence: Department of Statistics, Columbia University, 1255 Amsterdam Avenue, New York, NY 10027.}}


\markboth{IEEE TRANSACTIONS ON IMAGE PROCESSING,~Vol.~XX, No.~XX, XXXX~XXXX}%
{Wang \MakeLowercase{\textit{et al.}}: Automated and Robust Quantification of Colocalization in Dual-Color Fluorescence Microscopy}



\maketitle

\begin{abstract}
Colocalization is a powerful tool to study the interactions between fluorescently labeled molecules in biological fluorescence microscopy. However, existing techniques for colocalization analysis have not undergone continued development especially in regards to robust statistical support. In this paper, we examine two of the most popular quantification techniques for colocalization and argue that they could be improved upon using ideas from nonparametric statistics {  and scan statistics}. In particular, we propose a new colocalization metric that is robust, easily implementable, and optimal in a rigorous statistical testing framework. Application to several benchmark datasets, as well as biological examples, further demonstrates the usefulness of the proposed technique.
\end{abstract}


% \fi
\IEEEpeerreviewmaketitle

\section{Introduction}

\IEEEPARstart{C}{olocalization} is a powerful tool in examining macromolecules' spatial relationships to other macromolecules and cellular features. The goal of colocalization is to quantify the co-occurrence and/or correlation between two fluorescently-labeled molecules. Colocalization via fluorescence microscopy can yield quantitative, correlative spatiotemporal information. Yet historically, it has been often conducted in a rather ad hoc fashion, primarily through visual inspection of the overlaid microscopic images for both fluorescent signals; when two molecules of interest are labeled in ``red" and ``green", colocalization between them can be identified as ``yellow" in an overlaid image. As such, colocalization studies can be subject to misinterpretation and inconsistencies \citep[see, e.g.,][]{bolte2006guided,comeau2006guide}. To address this concern, numerous approaches have been proposed moving colocalization towards more rigorous and robust quantification \citep[see, e.g.,][among many others]{manders1992dynamics, manders1993measurement, costes2004automatic, ICCSN, dunn2011}.

Arguably the most widely-used quantitative measures for colocalization are Pearson's correlation coefficient and Manders' split coefficients. Pearson's correlation coefficient was first introduced to the microscopy community by \cite{manders1992dynamics}.  It measures the {\it linear relationship} of the intensities between the two channels, and a strong correlation indicates that a large intensity in one channel is often associated with a large intensity in the other. Another popular colocalization measure is the Manders' split coefficients proposed by \cite{manders1993measurement}.  These coefficients measure fractions of signal in one channel that overlap with the other.

Pearson's correlation coefficient and Manders' split coefficients measure the degree of colocalization manifested in two distinct ways: correlation and co-occurence, respectively. The former is most appropriate if two probes co-distribute proportionally to each other; whereas the latter is most useful if simple spatial overlap between the two probes is expected. We argue that both can be characterized as metrics of specific types of {\it positive dependence}. In statistical jargon, both Pearson's correlation coefficient and Manders' split coefficients are {\it parametric} in nature, which means that they work best when specific modeling assumptions hold; for example, Pearson's correlation works when the relationship between channels is linear. However, given the complexities that exist within biological contexts when measuring colocalization, this motivates us to consider a more robust method to quantify more general positive dependencies between two probes. {  To this end, we cast the colocalization analysis as a nonparametric statistical testing problem.} The approach we introduce {  for the testing problem} here is {  naturally} {\it nonparametric}, which works under much more general circumstances, as colocalization may display other types of associations beyond correlation or co-occurrence and may not be captured effectively by these two classical methods. {  The idea of nonparametric correlation coefficient in colocalization analysis has previously been introduced, e.g. \cite{adler2008, french2008}; however our work is the first to conduct colocalization analysis in a fashion of rigorous nonparametric statistical testing so that false discovery can be better controled and the value of coefficients can be transformed into statistical significance for easier interpretation.}

Not only is the biology itself adding complexity to colocalization analyses, but added complications are introduced during the acquisition process of biological samples, including varying background levels, leading to the need for extensive pre-processing before colocalization analyses can be applied.  When applying either Pearson's correlation coefficient or Manders' split coefficients to dual-channel fluorescence microscopic images directly, one might ignore an important fact that a dark background with positive offset may occupy a substantial area of the image. The power of either method critically depends on one's ability to determine an appropriate background level. Oftentimes, the solution is to avoid or exclude background pixels through manual selection a region of interest \citep{zinchuk2011quantifying,dunn2011}. More principled approaches have also been considered. In particular, global threshold reduction \citep[see, e.g.,][]{costes2004automatic} and local median threshold reduction \citep[see, e.g.,][]{ICCSN,zinchuk2011quantifying,dunn2011} have been widely used. In general, determining the background is a complex process and very susceptible to misspecification, as well as a lack of reproducibility.  There is a need for more robust colocalization analyses that can tease out the hidden, true biology without the need for user-based interference and manipulation via these pre-processing steps. The approach we developed here automatically adjusts for background, and therefore addresses this challenge in a seamless fashion.

In this paper, we discuss the main ideas behind different quantification techniques of colocalization and introduce our approach as a more general and robust alternative to those most frequently used. We provide a more rigorous justification of the proposed approach and show that the proposed colocalization score yields optimal test of colocalization under mild regularity conditions. Numerical experiments, both simulated and real, are also presented  to further demonstrate the merits of our proposed method.

\section{Robust Quantification of Colocalization}

To emphasize the need for a more robust quantification of colocalization, we first note that the usefulness of either Pearson's correlation coefficient or Manders' split coefficients relies on certain parametric assumptions about the data, albeit implicitly. 
Let $\II$ be the index set for all pixels in an image or a region of interest and, denoted by the pair $(X_i,Y_i)$, the intensity of the two channels measured at pixel $i\in \II$. Then the Pearson's correlation coefficient between two channels is given by
\begin{equation}
r={\sum_i(X_i-\bar{X})(Y_i-\bar{Y})\over \sqrt{\sum_i(X_i-\bar{X})^2\sum_i(Y_i-\bar{Y})^2}},
\end{equation}
where $\bar{X}$ and $\bar{Y}$ are the average intensities of the two channels, respectively. As mentioned previously, Pearson's correlation only measures the linear relationship of the intensities between two channels, and therefore may not be able to capture colocalization to its full extent. Consider a simple example where intensities $(X,Y)$ from the two channels can be modeled as a bivariate log-normal distribution. More concretely, $(\log(X),\log (Y))$ follows a bivariate normal distribution with mean $0$, variance $3^2=9$, and correlation coefficient $\rho$ so that the average intensity for each channel is approximately 90. The left panel of Figure \ref{fig:pcc-diff} gives the (population) Pearson's correlation coefficient between $X$ and $Y$ as a function of $\rho$ (i.e. the Pearson's correlation between $\log(X)$ and $\log (Y)$) and clearly shows that even very strong linear relationships on the log-scale may result in only modest Pearson correlation coefficients. In other words, Pearson's correlation is heavily influenced by nonlinear transformation on each channel. To further demonstrate this potential deficiency, $4\times 4$ images in two channels are given in the two right panels of Figure \ref{fig:pcc-diff}, whose intensities were generated from lognormal distribution with $\rho=0.9$. Despite the apparent colocalization between the two channels, both in terms of $\rho$ and visually, the Pearson's correlation coefficient is a mere 56\%.

\begin{figure}[htbp]
\centering
\includegraphics[width=3.7cm]{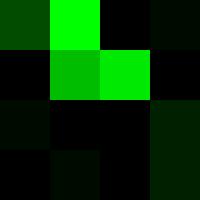}
\includegraphics[width=3.7cm]{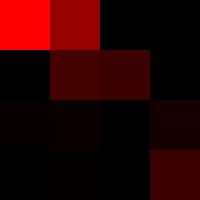}
\includegraphics[width=5cm]{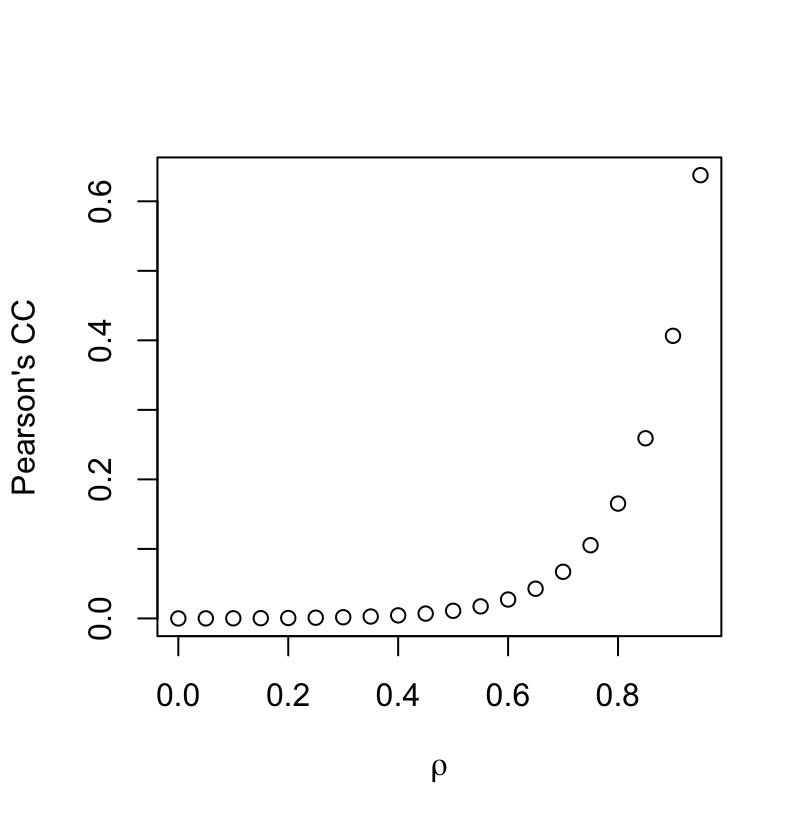}
\caption{Pearson's correlation coefficient is limited in that it only measures linear relationships between two channels. The lower panel shows the (population) correlation coefficient as a function of the correlation coefficient on the log-scale. The two upper panels  show sample images which, on the log-scale, have a correlation coefficient of 90\%; yet, on the original scale, it is only 56\%.}
\label{fig:pcc-diff}
\end{figure}

There are similar deficiencies for Manders' split coefficients as well. Specifically, Manders' split coefficients are defined by
$$
M_1={\sum_{i: Y_i>\alpha_Y} X_i\over \sum_{i\in \II} X_i}\qquad{\rm and}\qquad M_2={\sum_{i: X_i>\alpha_X} Y_i\over \sum_{i\in \II} Y_i},
$$
where the two thresholds $\alpha_X$ and $\alpha_Y$ are chosen appropriately so that any intensities below their respective threshold can be deemed as ``background''.
%the effectiveness of $M_1$ and $M_2$ in measuring colocalization depends critically on the choice of the thresholds $\alpha_X$ and $\alpha_Y$. They are often chosen subjectively, leading to ambiguity and inconsistency. A particularly common choice is advocated by \cite{costes2004automatic} where we set $\alpha_X$ and $\alpha_Y$ such that the pixels that are characterized as background, that is $X_i<\alpha_X$ and $Y_i<\alpha_Y$, have intensities with non-positive Pearson's correlation coefficient. Clearly, in the presence of noise, this may not perfectly separate signal from background so that $M_1$ and $M_2$ do not honestly reflect the amount of colocalization. 
It is worth noting that $M_1$ and $M_2$ can also be viewed as measures of the linear relationship between $X_i$s and ${\bf 1}(Y_i>\alpha_Y)$s, and $Y_i$s and ${\bf 1}(X_i>\alpha_X)$s, respectively, where ${\bf 1}(\cdot)$ is the indicator function. In other words, despite their differences in appearance, both Pearson's correlation and Manders' split coefficients can be viewed as measures for linear relationships between $X_i$s and $Y_i$s or their specific {\it monotonic transformations}. Motivated by this observation, we can consider a more general metric for dependence between $X_i$s and $Y_i$s under {\it arbitrary} monotonic transformations; more specifically, we have opted to quantify colocalization by Kendall's tau.

Let $n=|\II|$, the cardinality of $\II$. We call a pair of observations $(X_i, Y_i)$ and $(X_j, Y_j)$ ($i\neq j$) concordant if $\sign(X_i-X_j)\sign(Y_i-Y_j)>0$ and discordant if otherwise. Kendall tau for $\{(X_i,Y_i): i\in \II\}$ is then defined as the difference between the number of concordant pairs and discordant pairs divided by the total number of pairs, that is,
$$
\tau={1\over n(n-1)}\sum_{i\neq j}\sign(X_i-X_j)\sign(Y_i-Y_j).
$$
It is clear that $\tau$ depends on the data $\{(X_i,Y_i): i\in \II\}$ only through their ranks among $X_i$s and $Y_i$s so that it is invariant with respect to any monotonic transformations of $X_i$s and $Y_i$s.

As any other metric, when using $\tau$ to measure the degree of colocalization, it is essential to correct for background, and it may be fruitless to assess colocalization at locations where both channels are void of any real signal. To this end, it is of interest to evaluate Kendall tau only on the subset of pixels where both channels are sufficiently bright, leading to 
$$
\tau(t_X,t_Y)={\sum_{i,j\in \Kcal(t_X,t_Y):i\neq j}\sign(X_i-X_j)\sign(Y_i-Y_j)\over n_{t_X,t_Y}(n_{t_X,t_Y}-1)},
$$
where
$$
\Kcal(t_X,t_Y)=\left\{i\in \II:X_i\ge t_X,Y_i\ge t_Y\right\} 
$$
and
$$
n_{t_X,t_Y}=|\Kcal(t_X,t_Y)|,
$$
for two pre-specified $t_X$ and $t_Y$. We shall also adopt the convention that $\tau(t_X,t_Y)=-\infty$ if $n_{t_X,t_Y}\le 1$. {  }

Obviously, in practice, we do not know at which level $t_X$ and $t_Y$ colocalization may occur. To overcome this problem, we consider instead the maximum of normalized Kendall tau correlation for all possible $t_X$s and $t_Y$s. Note that the variance of $\tau(t_X,t_Y)$, when $X$ is independent from $Y$, is
$$
2(2n_{t_X,t_Y}+5)/9n_{t_X,t_Y}(n_{t_X,t_Y}-1).
$$
We shall therefore consider the following metric for colocalization
$$
\tau^\ast:=\max_{\substack{t_X\ge X_{(\lfloor n/2\rfloor)},\\ t_Y\ge Y_{(\lfloor n/2\rfloor)}}} \left\{\tau(t_X,t_Y)\cdot \sqrt{9n_{t_X,t_Y}(n_{t_X,t_Y}-1)\over 2(2n_{t_X,t_Y}+5)}\right\},
$$
where $X_{(k)}$ and $Y_{(k)}$ are the $k$th order statistics of $X_i$s and $Y_i$s respectively. Note that the lower bounds $X_{(\lfloor n/2\rfloor)}$ and $Y_{(\lfloor n/2\rfloor)}$ are chosen for convenience and can be replaced by other values. In particular, they can be taken as approximated thresholds of signal so that only possible thresholds above those approximated are considered.

The nonparametric version colocalization measure $\tau^\ast$ is more robust in at least two ways, compared to Pearson's correlation coefficient or Manders' split coefficients. $\tau^\ast$ is invariant with respect to arbitrary monotonic transformations of $X$ and $Y$. Furthermore, $\tau^\ast$ only takes real correlation on signal into account and is immune from the presence of background. To demonstrate the merit of $\tau^\ast$, we discuss the theoretical properties of $\tau^\ast$ under appropriate models in the next section.

{  {\it Remark:} The non-parametric correlation coefficient can reflect the general associations between variables in a more precise way, compared with parametric correlation coefficient \citep[see, e.g.][]{dengler2010,embrechts2002}. To illustrate this, we consider two examples to compare the Kendall tau correlation coefficient $\tau$ we use here and one of the most widely used correlation coefficient, Pearson correlation coefficient $r$. The first example is when $X$ and $Y$ are drawn from independent $t$-distributions with degrees of freedom less than 4. In this case, the variance of Pearson's correlation coefficient $r$ is not well defined \citep[see, e.g.][]{dengler2010}, so that $r$ might deviate from 0 with large probability. On the other hand, Kendall tau correlation converges to 0 as the sample size increases, as long as $X$ and $Y$ are independent, immune from any heavy tail distributions. In the second example \citep[see, e.g.][]{embrechts2002}, $X$ is drawn from a log-normal distribution and $Y=X^S$ for some integer $S$.  The Pearson's correlation coefficient between $X$ and $Y$ is 
$$
{e^S-1\over \sqrt{(e-1)(e^{S^2}-1)}}\to 0,\quad {\rm as}\ S\to \infty,
$$ 
despite the fact that $Y$ is totally determined by $X$. However, the Kendall tau correlation between $X$ and $Y$ is always 1, reflecting the strong connection between $X$ and $Y$. This suggests Kendall tau correlation $\tau$ is able to capture a wider range of association than Pearson's correlation coefficient $r$. Therefore, Kendall's tau $\tau$ can reflect correlation more precisely than Pearson's correlation coefficient $r$.}

%{  {\it Remark:} The principle of truncation and scanning above can also be applied, as a framework, to other colocalization measures. For example, the truncated Pearson's correlation coefficient is given accordingly
%$$
%r(t_X,t_Y)={\sum_{i\in \Kcal(t_X,t_Y)}(X_i-\bar{X})(Y_i-\bar{Y})\over \sqrt{\sum_{i\in \Kcal(t_X,t_Y)}(X_i-\bar{X})^2\sum_{i\in \Kcal(t_X,t_Y)}(Y_i-\bar{Y})^2}},
%$$
%where $\bar{X}$ and $\bar{Y}$ are the average intensities in region $\Kcal(t_X,t_Y)$ of the two channels. The maximum of normalized Pearson's correlation coefficient can be defined in a similar fashion with $\tau^\ast$:
%$$
%r^\ast:=\max_{\substack{t_X\ge X_{(\lfloor n/2\rfloor)},\\ t_Y\ge Y_{(\lfloor n/2\rfloor)}}} \left\{r(t_X,t_Y)\cdot \sqrt{n_{t_X,t_Y}}\right\}.
%$$
%Like $\tau^\ast$, $r^\ast$ is able to adjust for the background automatically as well.}

\section{Statistical Significance}
\label{sc:opt}

To translate the proposed metric for colocalization $\tau^\ast$ into statistical significance, we now consider a hypothesis testing framework for colocalization. To this end, $F$ denotes the joint distribution function for the pair $(X_i,Y_i)$, where $i=1,\ldots, n$. In the absence of colocalization (null hypothesis, $H_0$), the two channels can be expected to be behave independently so that
\begin{equation}
\label{eq:condnull}
H_0: F(x,y)=F_X(x)F_Y(y), \qquad \forall x,y\in \RR,
\end{equation}
where $F_X(x)=F(x,+\infty)$ and $F_Y(y)=F(+\infty,y)$ are the marginal distribution functions. On the other hand, in the presence of colocalization (alternative hypothesis, $H_1$), we expect that $X$ and $Y$ are positively dependent. Furthermore, positive dependency only applies to signals; that is, there exists some $\eta_X$ and $\eta_Y$ such that the conditional distribution of $(X,Y)$ given that $X>\eta_X$ and $Y>\eta_Y$, hereafter denoted by $F_{\eta_X,\eta_Y}$, is {\it positively quadrant dependent}. Specifically, if $F_{\eta_X,\eta_Y}$ is positively quadrant dependent, then $F_{\eta_X,\eta_Y}(x,y)\ge F_{\eta_X}(x)F_{\eta_Y}(y)$ for all $x,y\in \RR$, and there exist $x,y\in \RR$ such that $F_{\eta_X,\eta_Y}(x,y)> F_{\eta_X}(x)F_{\eta_Y}(y)$ where $F_{\eta_X}(x)=F_{\eta_X,\eta_Y}(x,+\infty)$ and $F_{\eta_Y}(y)=F_{\eta_X,\eta_Y}(+\infty,y)$ are the marginal distributions of $F_{\eta_X,\eta_Y}$ \citep[see, e.g.,][]{lehmann1966,nelsen2006}. We do not assume prior knowledge of $\eta_X$ and $\eta_Y$ so that
$$
H_1: \exists\ \eta_X, \eta_Y {\rm \ s.t.}\ F_{\eta_X,\eta_Y} {\rm \ is\ positively\ quadrant\ dependent.}
$$
The colocalization metric $\tau^\ast$ can be used to effectively test $H_0$ against $H_1$ and therefore can be converted into p-values as a scale-free measure of colocalization, which we will discuss in more detail in the next section.

\subsection{Optimality}
As previously stated, the colocalization metric $\tau^\ast$ provides an efficient statistic for testing $H_0$ against $H_1$. To this end, let $q_\alpha$ denote the $1-\alpha$ quantile of the distribution of $\tau^\ast$ under $H_0$. Although there is no closed-form analytic expression for $\tau^\ast$, it can be readily evaluated by Monte Carlo schemes. We shall discuss in further details practical issues of implementation in the next subsection. Once $q_\alpha$ is computed, we can then proceed to reject $H_0$ and therefore claim colocalization as soon as the observed $\tau^\ast$ is greater than $q_\alpha$. We denote this test by $\Delta$. It is clear that $\Delta$ is an $\alpha$ level test; we now argue that it is also optimal in the sense that it can detect evidence of colocalization at a level that no other tests could improve.

Note first that positive quadrant dependence of $F_{\eta_X,\eta_Y}$ immediately implies that for two independent copies $(X,Y)$ and $(\tilde{X},\tilde{Y})$ following distribution $F$,
\begin{align*}
&T(\eta_X,\eta_Y):=\\
&\PP\left\{(X-\tilde{X})(Y-\tilde{Y})>0|X,\tilde{X}>\eta_X; Y,\tilde{Y}>\eta_Y\right\}-\\
&\PP\left\{(X-\tilde{X})(Y-\tilde{Y})<0|X,\tilde{X}>\eta_X; Y,\tilde{Y}>\eta_Y\right\}>0.
\end{align*}
%\begin{eqnarray*}
%T(\eta_X,\eta_Y):=\EE\left\{{\rm sign}(X-\tilde{X}){\rm sign}(Y-\tilde{Y})|X,\tilde{X}>\eta_X; Y,\tilde{Y}>\eta_Y\right\}.
%\end{eqnarray*}
In other words, under null hypothesis $H_0$, $T(\eta_X,\eta_Y)=0$, for all $\eta_X,\eta_Y$; while under the alternative hypothesis $H_1$,
$$\sup_{\eta_X,\eta_Y}T(\eta_X,\eta_Y)>0.$$

\begin{theorem}
\label{thm:alter}
Assume that $\{(X_i,Y_i): i\in \II\}$ ($n:=|\II|$) are independently sampled from $F$ obeying
\begin{equation}
\label{eq:altercond}
\sup_{\eta_X,\eta_Y} V(\eta_X,\eta_Y)\cdot T^2(\eta_X,\eta_Y) \gg {\log\log n\over n}.
\end{equation}
Here $V(\eta_X,\eta_Y):=1+F(\eta_X,\eta_Y)-F_X(\eta_X)-F_Y(\eta_Y)$. Then $\Delta$ is a consistent test in that we reject $H_0$ in favor of $H_1$ with probability tending to one. Conversely, there exists a constant $c>0$ such that for any $\alpha$-level test $\Delta$ based on sample $\{(X_i,Y_i): i\in \II\}$, there is  an instance where joint distribution function $F$ obeying
\begin{equation}
\label{eq:lowercond}
\sup_{\eta_X,\eta_Y} V(\eta_X,\eta_Y)\cdot T^2(\eta_X,\eta_Y) \ge c{\log\log n\over n}
\end{equation}
and yet, we accept $H_0$ with probability tending to $1-\alpha$ as if $H_0$ holds.
\end{theorem}

Hereafter, we write $a_n\gg b_n$ if $b_n=o(a_n)$. Theorem \ref{thm:alter} provides theoretical justifications that $\tau^\ast$ is an appropriate and powerful test statistic for $H_0$ against $H_1$. In particular, it suggests that $\tau^\ast$ is optimal in the sense that it can detect correlation at a level no other tests could significantly improve.

\subsection{Practical Considerations}
\label{sc:fast}

In practice, it is more useful to report the p-value associated with an observed $\tau^\ast$ rather than just a simple decision on rejecting or accepting the null hypothesis $H_0$. To this end, we can compare the observed $\tau^\ast$ from a dual-channel microscopic image with the sampling distribution of $\tau^\ast$ when there is no colocalization. We can apply a permutation test to estimate the sampling distribution of $\tau^\ast$ under $H_0$. More specifically, we can randomly shuffle $\{X_i: i\in \II\}$ or $\{Y_i: i\in \II\}$. Random arrangement ensures that there is no meaningful colocalization between the two channels. For each shuffled or permuted sample, we recompute $\tau^\ast$. The null distribution of $\tau^\ast$ can therefore be estimated by repeating the random rearrangement many times. 

When implementing this strategy, there are two practical challenges. The first potential hurdle is the computational cost. It is not hard to see that
\begin{equation}
\label{eq:full}
\tau^\ast=\max_{\substack{t_X=X_{(j)},\\ t_Y=Y_{(k)}:\\j,k\ge \lfloor n/2\rfloor}} \left\{\tau(t_X,t_Y)\cdot \sqrt{9n_{t_X,t_Y}(n_{t_X,t_Y}-1)\over 2(2n_{t_X,t_Y}+5)}\right\}.
\end{equation}
There are a total of $O(n^2)$ possible pairs of $(j,k)$, and fast evaluation of Kendall tau requires $O(n\log n)$ floating-point operations.
Thus, the exact computation of $\tau^\ast$ has complexity $O(n^3\log n)$. This could be quite expensive to compute for even a moderately-sized image, and particularly so because we need to compute $\tau^\ast$ for many scrambled images.

To this end, we propose to compute an approximation of $\tau^\ast$. More specifically, instead of evaluating the maximum over $O(n^2)$ possible pairs of $(j,k)$ as in \eqref{eq:full}, we consider the maximum over only a subset of these pairs. Let
\begin{align*}
\Rcal_n:=\left\{s:s=\left\lfloor n-\left(1+{1\over \log\log n}\right)^j\right\rfloor, j\in\mathbb{N}_{+}, s\ge \lfloor n/2\rfloor \right\}.
\end{align*}
Here $\mathbb{N}_{+}$ refers to the set of all positive integers. In other words, $\Rcal_n$ is a collection of coordinates that are nearly a geometric series. As such, the number of pairs in $\Rcal_n$ is much smaller than the original ones, as illustrated in Figure \ref{fig:fast}.
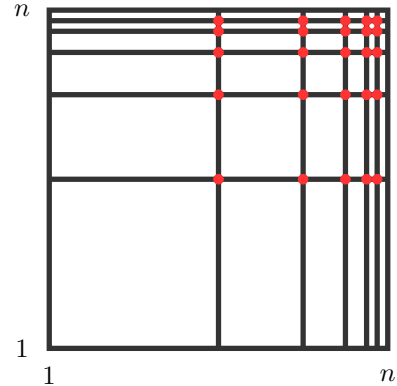
\begin{figure}[h]
\begin{center}
\begin{tikzpicture}[scale=0.9]
\draw[line width=2pt,black!80!white] (0,0) rectangle (5,5);
\path[draw,line width=2pt,black!80!white] (2.5,0) -- (2.5,5);
\path[draw,line width=2pt,black!80!white] (3.75,0) -- (3.75,5);
\path[draw,line width=2pt,black!80!white] (4.375,0) -- (4.375,5);
\path[draw,line width=2pt,black!80!white] (4.6875,0) -- (4.6875,5);
\path[draw,line width=2pt,black!80!white] (4.84375,0) -- (4.84375,5);
\path[draw,line width=2pt,black!80!white] (0,2.5) -- (5,2.5);
\path[draw,line width=2pt,black!80!white] (0,3.75) -- (5,3.75);
\path[draw,line width=2pt,black!80!white] (0,4.375) -- (5,4.375);
\path[draw,line width=2pt,black!80!white] (0,4.6875) -- (5,4.6875);
\path[draw,line width=2pt,black!80!white] (0,4.84375) -- (5,4.84375);
\filldraw[red!80!white] (2.5,2.5) circle (0.07);
\filldraw[red!80!white] (2.5,3.75) circle (0.07);
\filldraw[red!80!white] (2.5,4.375) circle (0.07);
\filldraw[red!80!white] (2.5,4.6875) circle (0.07);
\filldraw[red!80!white] (2.5,4.84375) circle (0.07);
\filldraw[red!80!white] (3.75,2.5) circle (0.07);
\filldraw[red!80!white] (3.75,3.75) circle (0.07);
\filldraw[red!80!white] (3.75,4.375) circle (0.07);
\filldraw[red!80!white] (3.75,4.6875) circle (0.07);
\filldraw[red!80!white] (3.75,4.84375) circle (0.07);
\filldraw[red!80!white] (4.375,2.5) circle (0.07);
\filldraw[red!80!white] (4.375,3.75) circle (0.07);
\filldraw[red!80!white] (4.375,4.375) circle (0.07);
\filldraw[red!80!white] (4.375,4.6875) circle (0.07);
\filldraw[red!80!white] (4.375,4.84375) circle (0.07);
\filldraw[red!80!white] (4.6875,2.5) circle (0.07);
\filldraw[red!80!white] (4.6875,3.75) circle (0.07);
\filldraw[red!80!white] (4.6875,4.375) circle (0.07);
\filldraw[red!80!white] (4.6875,4.6875) circle (0.07);
\filldraw[red!80!white] (4.6875,4.84375) circle (0.07);
\filldraw[red!80!white] (4.84375,2.5) circle (0.07);
\filldraw[red!80!white] (4.84375,3.75) circle (0.07);
\filldraw[red!80!white] (4.84375,4.375) circle (0.07);
\filldraw[red!80!white] (4.84375,4.6875) circle (0.07);
\filldraw[red!80!white] (4.84375,4.84375) circle (0.07);
\draw[thick,black](0,-0.4)node{$1$};
\draw[thick,black](5,-0.4)node{$n$};
\draw[thick,black](-0.4,0)node{$1$};
\draw[thick,black](-0.4,5)node{$n$};
\end{tikzpicture}
\end{center}
\caption{All possible pairs $(j,k)$ when both $j$ and $k$ are in $\Rcal_n$}
\label{fig:fast}
\end{figure}

We then define
\begin{equation}
\label{eq:fast}
\tau_{\rm app}^\ast:=\max_{\substack{t_X=X_{(j)},\\ t_Y=Y_{(k)}:\\ j,k\in \Rcal_n}} \left\{\tau(t_X,t_Y)\cdot \sqrt{9n_{t_X,t_Y}(n_{t_X,t_Y}-1)\over 2(2n_{t_X,t_Y}+5)}\right\}.
\end{equation}
A careful inspection of the proof of Theorem \ref{thm:alter} shows that a test that uses $\tau_{\rm app}^\ast$ in place of $\tau^\ast$ remains optimal and consistent under condition \eqref{eq:altercond}. The idea of evaluating a statistic on an approximation set only to reduce the computation cost while retaining statistical power is commonly used in scan statistics \citep[see, e.g.,][]{arias2005near, walther2010optimal, chan2013detection, rivera2013optimal, wang2016structured}. {  The fast $\tau_{\rm app}^\ast$ can be applied to large scale microscopic images, as its computational complexity is almost linear with the number of pixels.}

Another practical challenge is the potential dependence among $X_i$s and $Y_i$s. The range of dependence within either channel is often determined by the numerical aperture of the objective lens, and the fluorescence emission wavelength, as shown previously by \cite{costes2004automatic}. It is important that we preserve such a dependence structure when estimating the sampling distribution of $\tau^\ast$. To this end, we can adopt the strategy advocated by \cite{costes2004automatic}; instead of scrambling the image pixel-by-pixel, we can divide the image into blocks with the number of pixels in each block determined by the point spread function and then scramble the image block-by-block.

With these two adjustments, we are ready to show the whole flow of our new method. {  In Algorithm~\ref{ag:new}, the input image can be an image before or after pre-processing. According to our experience, our method works very well on both raw images (see Section~\ref{sc:real}) and pre-processed images (see Section~\ref{sc:bench}). It is also worth noting that the $p$-value obtained in Algorithm~\ref{ag:new} is only calculated for a single experiment. Multiple comparison correction is needed if we apply Algorithm~\ref{ag:new} on multiple images.} Algorithm~\ref{ag:new} has been implemented in \texttt{R} package \texttt{RKColocal}, which is openly available (see \url{https://github.com/lakerwsl/RKColocal}). 
\begin{algorithm}
\caption{our new method based on $\tau^\ast$ (or $\tau_{\rm app}^\ast$)}
\label{ag:new}
\begin{algorithmic}
\REQUIRE channel intensities $\{X_i\}_{i\in \II}$, $\{Y_i\}_{i\in \II}$, repeating times $B$ and block size $D$
\ENSURE $p$-value
\STATE $E_0$ $\leftarrow$ calculate $\tau^\ast$ (or $\tau_{\rm app}^\ast$) on $\{X_i\}_{i\in \II}$ and $\{Y_i\}_{i\in \II}$.
\FOR{$j=1$ to $B$}
\STATE  $\{\tilde{X}_i\}_{i\in \II}$ $\leftarrow$ block-wise randomly shuffle $\{X_i\}_{i\in \II}$ with block size $D$
\STATE  $E_j$ $\leftarrow$ calculate $\tau^\ast$ (or $\tau_{\rm app}^\ast$) on $\{\tilde{X}_i\}_{i\in \II}$ and $\{Y_i\}_{i\in \II}$
\ENDFOR
\STATE $P$ $\leftarrow$ $\#\{E_j>E_0\}/B$
\RETURN $P$
\end{algorithmic}
\end{algorithm}

\section{Numerical Experiments}
\label{sc:num}

To demonstrate the merits of our proposed method, we conducted several sets of numerical experiments, applying our method on both simulated and biological image data.

\subsection{Simulated Data Examples}

To simulate the positive dependence between the two channels, we consider a setting based on Clayton copula \citep[see, e.g.,][]{nelsen2006}. More specifically, under the null hypothesis $H_0$ (no colocalization), we simulated the intensities of each pixel $X$ and $Y$ according to
\begin{equation}
\label{eq:monotran}
X=e^{8(U-0.5)}\quad {\rm and }\quad Y=e^{8(V-0.5)},
\end{equation}
where $U$ and $V$ are independently drawn from a uniform distribution between 0 and 1, $Unif([0,1])$. The image is blurred by applying gaussian smoothing (point-spread function (PSF) is gaussian kernel) after intensities of each pixel are simulated following the rule above. 
A typical example of a simulated dual channel image without colocalization is shown in the left most column of Figure~\ref{fg:exmpsim}.
To generate colocalization under the alternative hypothesis $H_1$, we first simulated bivariate random variables $(U,V)$ from a distribution:
$$
{d^2F(u,v)\over dudv}=\begin{cases}g_{\theta}\left({u-R\over 1-R},{v-R\over 1-R}\right)& (u,v)\in [R,1]\times [R,1] \\ 1 & (u,v)\in [0,1]^2\setminus [R,1]^2 \end{cases}
$$
where $g_{\theta}(u,v)$, $0<\theta<\infty$, is the density function of Clayton copula distribution, that is
\begin{align*}
g_{\theta}(u,v)&={d^2\over dudv}(u^{-\theta}+v^{-\theta}-1)^{-1/\theta}\\
&=(\theta+1)(uv)^{-(\theta+1)}(u^{-\theta}+v^{-\theta}-1)^{-(2\theta+1)/\theta}.
\end{align*}
Here, $R$ is a parameter between 0 and 1, representing a threshold above which colocalization occurs, as positive quadrature dependence occurs when $U,V>R$. A larger $R$ suggests colocalization occurs with less signal, so that detection of the colocalization is more difficult (compared to the second and third column in Figure~\ref{fg:exmpsim}). Another parameter, $\theta$, is a number larger than 0, controlling the the dependence/colocalization level above the thresholds $R$. Specifically, the degree of positive quadrature dependence when $U,V>R$ is $\theta/(\theta+2)$ i.e. $T(R,R)=\theta/(\theta+2)$. Thus, a larger $\theta$ implies higher correlation among the given signal (compared to the second and fourth column in Figure~\ref{fg:exmpsim}). The pixel intensities $(X,Y)$ follow the same monotone transformation of $(U,V)$ in (\ref{eq:monotran}), and the image is also blurred by gaussian smoothing. The three right images of Figure~\ref{fg:exmpsim} show examples of dual channel images with varying colocalization.
\begin{figure}
    \centering
        \begin{tikzpicture}[scale=1]
  \node[anchor=south west,inner sep=0] at (0,0) {\includegraphics[width=0.12\textwidth]{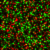}};
  \node[anchor=north west,inner sep=0] at (0,0) {\includegraphics[width=0.12\textwidth]{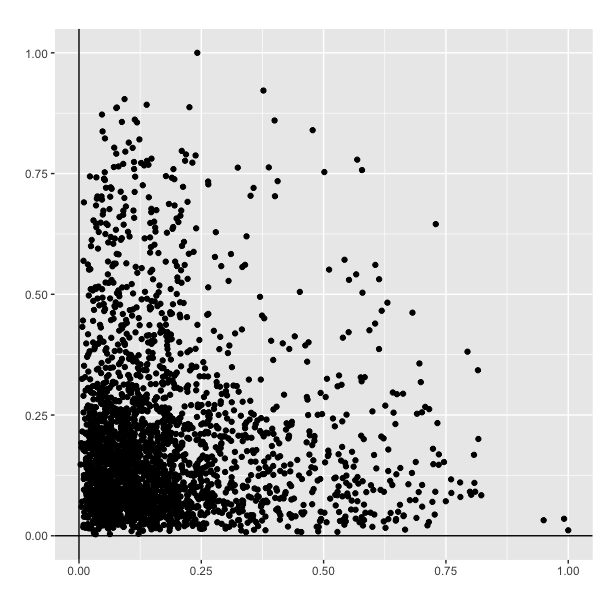}};
  \node[anchor=south west,inner sep=0] at (2.25,0) {\includegraphics[width=0.12\textwidth]{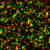}};
  \node[anchor=north west,inner sep=0] at (2.25,0) {\includegraphics[width=0.12\textwidth]{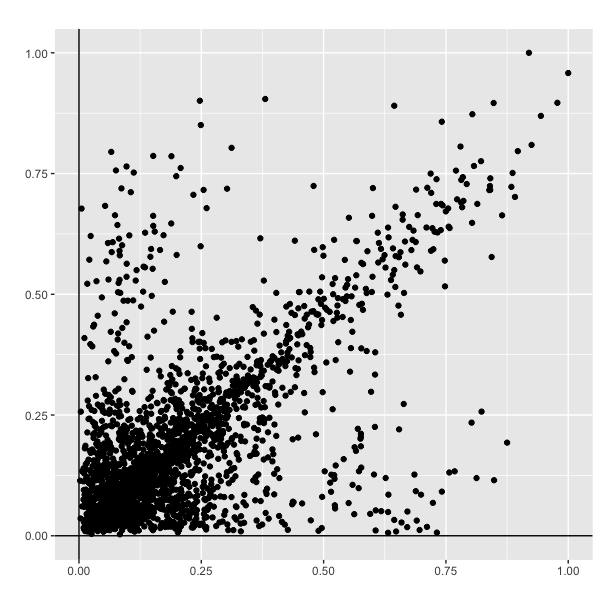}};
  \node[anchor=south west,inner sep=0] at (4.5,0) {\includegraphics[width=0.12\textwidth]{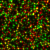}};
  \node[anchor=north west,inner sep=0] at (4.5,0) {\includegraphics[width=0.12\textwidth]{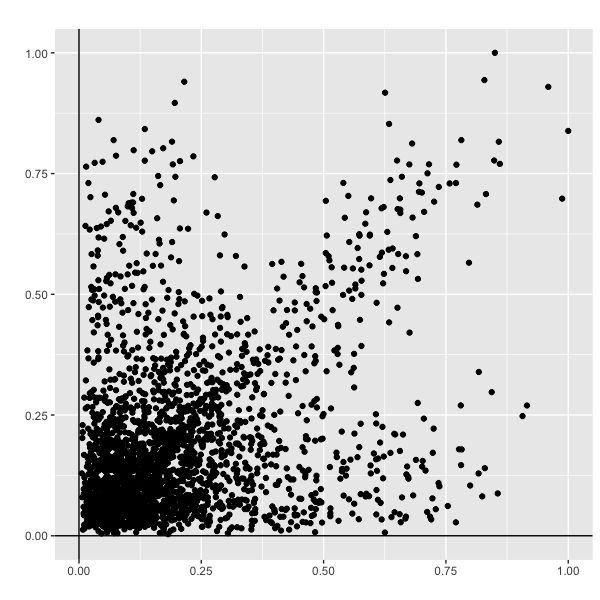}};
  \node[anchor=south west,inner sep=0] at (6.75,0) {\includegraphics[width=0.12\textwidth]{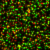}};
  \node[anchor=north west,inner sep=0] at (6.75,0) {\includegraphics[width=0.12\textwidth]{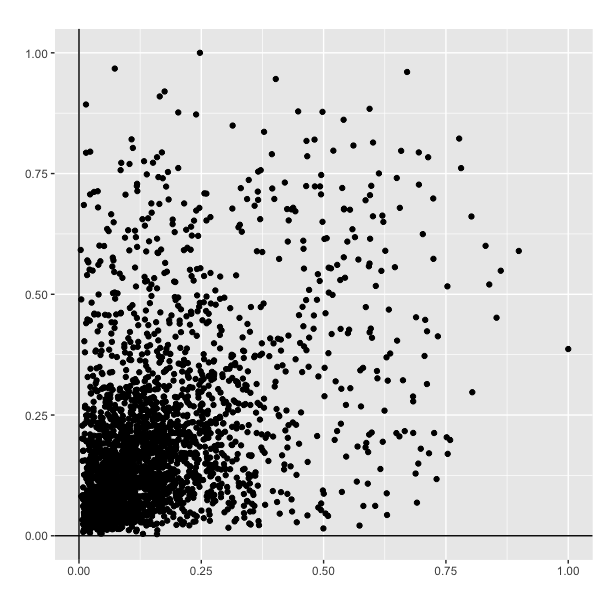}};
    \end{tikzpicture}
    \caption{Example images of simulated dual-channel data and their corresponding scatter plots (image size: $50\times 50$). From left to right: an example simulated dual-channel image without colocalization (under the null hypothesis $H_0$), with colocalization when $R=0.3$ and $\theta=100$ (under the alternative hypothesis $H_1$), with colocalization when $R=0.5$ and $\theta=100$ (under the alternative hypothesis $H_1$), and with colocalization when $R=0.3$ and $\theta=5$ (under the alternative hypothesis $H_1$).}\label{fg:exmpsim}
\end{figure}

In this simulation experiment, we compared our new method with the traditional colocalization quantitative measures, including Pearson's correlation coefficient $r$ and Manders' split coefficients $M_1,M_2$. In $M_1$ and $M_2$, the thresholds $\alpha_X$ and $\alpha_Y$ were chosen by applying Otsu's method to each channel. To make a comparison possible, we employed a statistical hypothesis testing framework and reported the decision associated with each quantitative measure.
Specifically, we simulated the null distribution\footnote{Null distribution is the distribution of the test statistic, i.e. $r$, $M_1$, $M_2$, or $\tau_{\rm app}^\ast$, under the null hypothesis $H_0$.} of colocalization quantitative measures, $r$, $M_1$, $M_2$ or $\tau_{\rm app}^\ast$, and identified the upper 5\% quantile of the null distribution as the critical value based on 1000 Monte Carlo simulations. In this way, we can ensure that the Type I error (the probability of false discovery) is controlled at level 5\% up to Monte Carlo simulation error. 
The reported decision rejects the null hypothesis if the corresponding colocalization quantitative measure exceeded its respective critical value, failing to reject the null hypothesis otherwise. 
Under this statistical hypothesis testing framework, the performance of colocalization quantitative measures can be assessed through the power of testing, i.e. the probability of rejecting the null hypothesis under the alternative hypothesis $H_1$. In this simulation study, the power $\beta$ is estimated by the proportion of null hypothesis rejection, i.e.
\begin{equation}
\label{eq:powerdef}
\beta={\rm number\ of\ null\ hypothesis\ rejection\over\rm number\ of\ simulation\ runs}.
\end{equation}
Clearly, a larger power $\beta$ suggests the colocalization measure is more efficient in colocalization detection.

To investigate the performance of different colocalization measures, we compare their power $\beta$ defined in (\ref{eq:powerdef}) when data is generated according to the alternative hypothesis model (colocalization exists) under different values of $R$ and $\theta$. We conducted the simulation experiments by varying parameters $R$ and $\theta$ in the alternative hypothesis model simultaneously. Specifically, we considered different values of $R$: 0.7, 0.8, and 0.9 and a range of $\theta$ from 1 to 10. For each combination of $R$ and $\theta$, we repeated the experiment 1000 times. In each experimental run, we simulated colocalized data on a $50\times50$ lattice and applied tests of $r$, $M_1$, $M_2$, or $\tau_{\rm app}^\ast$ on the simulated data. The decision of each colocalization measure was recorded and the power $\beta$ in 1000 experiments was calculated by (\ref{eq:powerdef}). The results of power $\beta$ are summarized in Figure~\ref{fg:fgm}. In Figure~\ref{fg:fgm}, the power $\beta$ of all methods increases along with $\theta$ increasing and $R$ decreasing, which is consistent with our discussion in the simulation setting introduction. These results show that the power $\beta$ of our new method is larger than that of Pearson's correlation coefficient and Manders' split coefficients at most $R$ and $\theta$, especially when there is less colocalized signal (i.e. $R$ is large). Therefore, we can conclude that $\tau^\ast_{\rm app}$ out-performs $r$, $M_1$, and $M_2$.

\begin{figure}[h]
\begin{center}
\includegraphics[width=0.5\textwidth]{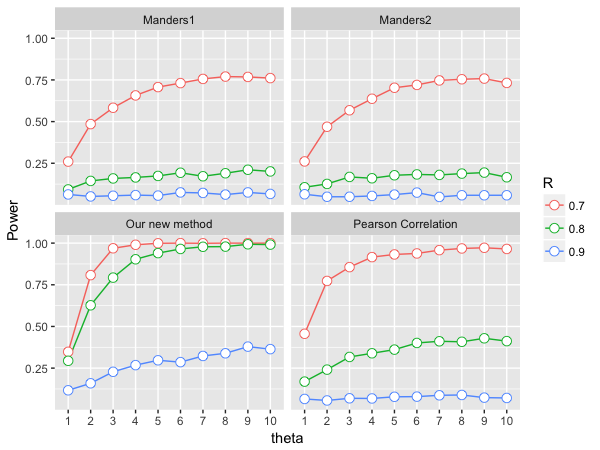}
\caption{The power comparison among colocalization quantitative measures under Clayton copula model. In each plot, the $x$ axis is the value of $\theta$ and $y$ axis is the value of power $\beta$ (between 0 and 1). Different colors of curve represent different values of $R$.}
\label{fg:fgm}
\end{center}
\end{figure}

\subsection{Benchmark Real Data Examples}
\label{sc:bench}

Next, we applied our new method to several benchmark real data examples from \cite{ICCSN}. The first example detected colocalization between the the ryanodine receptor (RyR) and the estrogen receptor alpha (ER$\alpha$) in a mouse heart cell (in Figure~\ref{fg:Wu1}). As described in \cite{ICCSN}, there is no evidence that these two proteins interact. The second example compared the distribution of RyR and $\alpha$1C calcium channel ($\alpha$1C) in a mouse cell, which are known to colocalize (in Figure~\ref{fg:Wu2}). The third example measured the behavior of the $\alpha$-subunit of Ca$^{2+}$ and voltage-dependent large conductance Kþ channels (MaxiK-$\alpha$) and that of $\alpha$-tubulin (in Figure~\ref{fg:Wu3}). These two types of proteins are partially colocalized according to \cite{ICCSN}.

For each of the three examples, the proposed metric $\tau^\ast_{\rm app}$ and the histogram of its null distributions obtained via block-wise permutations are given in Figure~\ref{fg:Wu}. In these experiments, the number of permutations is 1000 and block size is $\lfloor\min(\sqrt{a},\sqrt{b})\rfloor$ when the size of image is $a\times b$. These results are fairly consistent with those reported in \cite{ICCSN}. It is worth noting that \cite{ICCSN} also ran many existing methods, including Pearson correlation coefficient and Manders' split coefficients, on these data examples and concluded that such quantification methods are prone to false discovery. In particular, both Pearson correlation coefficient and Manders' split coefficients identified colocalization in the first example \citep[see][]{ICCSN}, contrary to the biology behind it (in Figure~\ref{fg:Wu1}).

\begin{figure*}
    \centering
    \begin{subfigure}[b]{0.23\textwidth}
    \centering
        \begin{tikzpicture}[scale=0.9]
  \node[anchor=south west,inner sep=0] at (0,0) {\includegraphics[width=\textwidth]{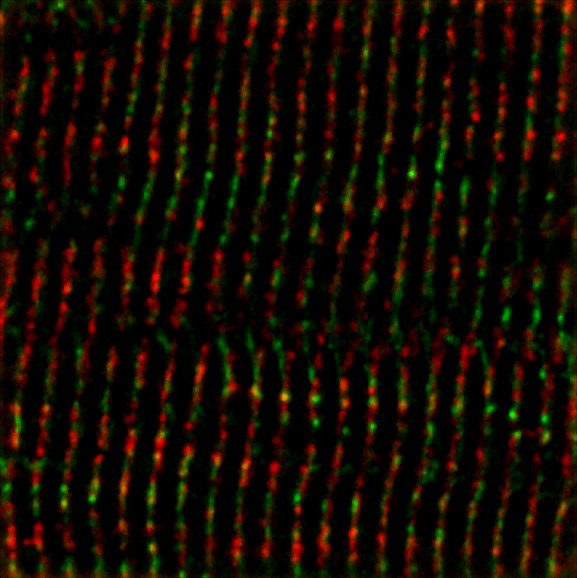}};
  \node[anchor=south west,inner sep=0] at (0,-2.5) {\includegraphics[width=\textwidth]{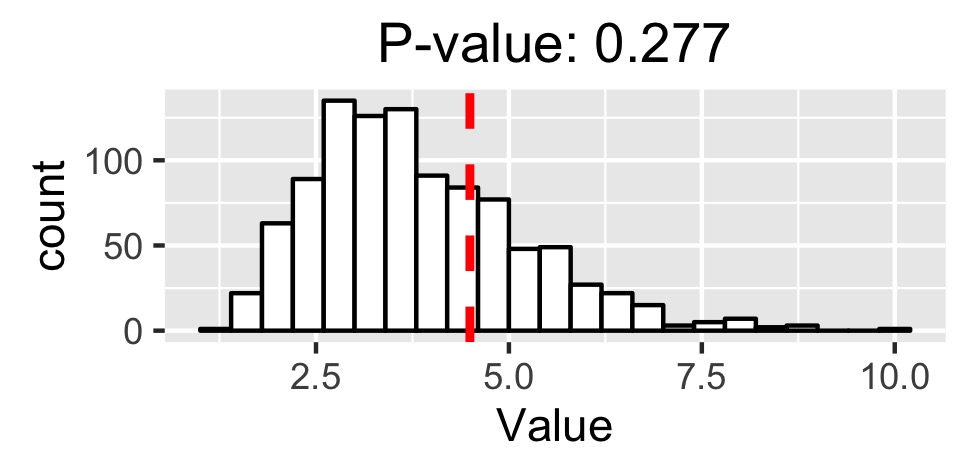}};
    \end{tikzpicture}
        \caption{Poor colocalization is expected between RyR and ER$\alpha$. (size: $577\times 578$)}
        \label{fg:Wu1}
    \end{subfigure}
    \hspace{0.05\textwidth}
    \begin{subfigure}[b]{0.23\textwidth}
    \centering
        \begin{tikzpicture}[scale=0.9]
  \node[anchor=south west,inner sep=0] at (0,0) {\includegraphics[width=\textwidth]{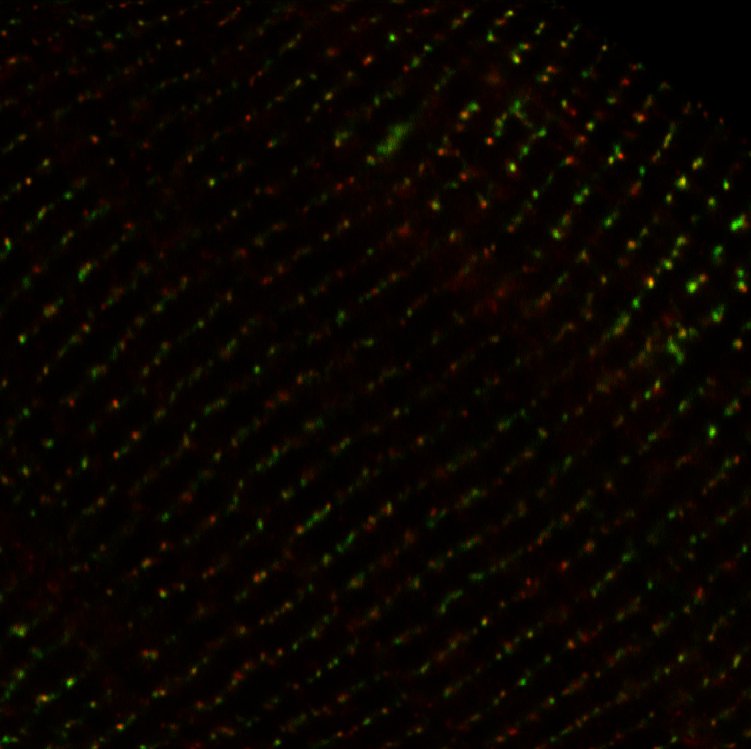}};
  \node[anchor=south west,inner sep=0] at (0,-2.5) {\includegraphics[width=\textwidth]{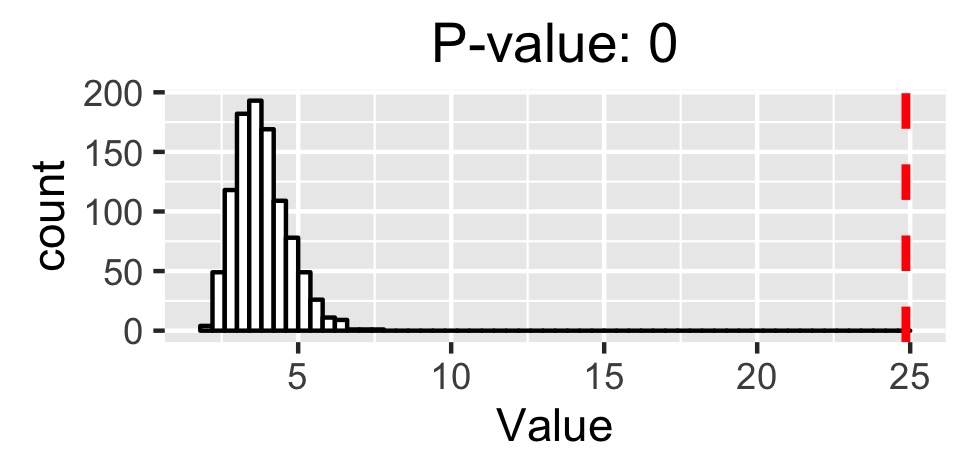}};
    \end{tikzpicture}
        \caption{Partial colocalization is expected between RyR and $\alpha$1C. (size: $751\times 749$)}
        \label{fg:Wu2}
    \end{subfigure}
    \hspace{0.05\textwidth}
    \begin{subfigure}[b]{0.23\textwidth}
    \centering
        \begin{tikzpicture}[scale=0.9]
  \node[anchor=south west,inner sep=0] at (0,0) {\includegraphics[width=\textwidth]{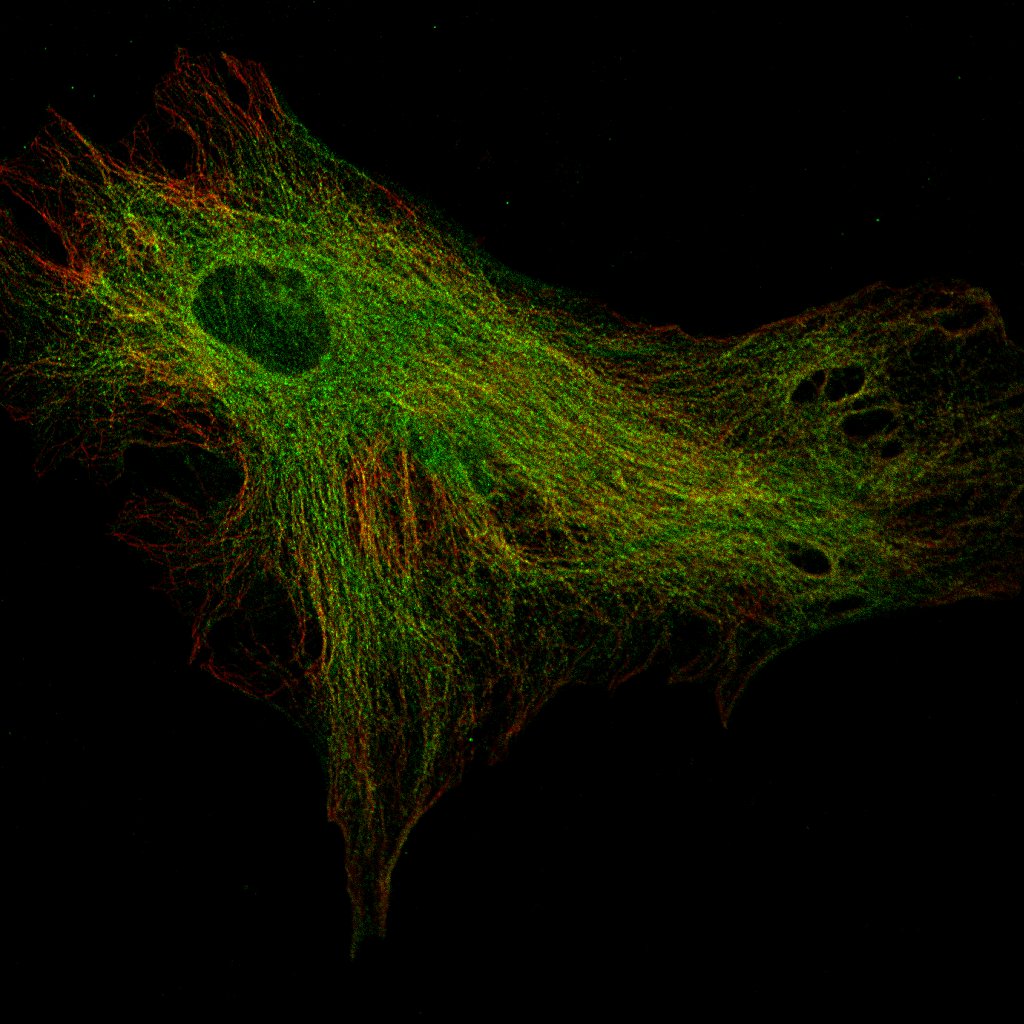}};
  \node[anchor=south west,inner sep=0] at (0,-2.5) {\includegraphics[width=\textwidth]{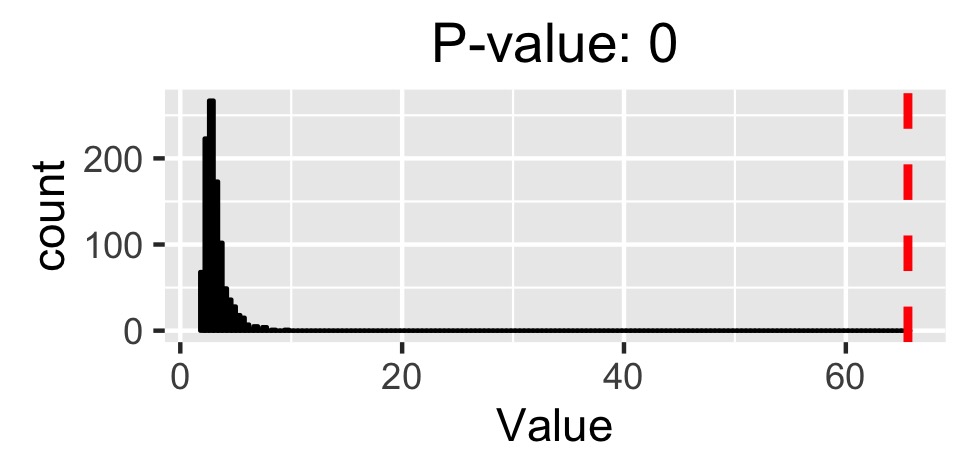}};
    \end{tikzpicture}
        \caption{ Partial colocalization is expected bewteen Ca$^{2+}$ and MaxiK-$\alpha$. (size: $1024\times 1024$)}
        \label{fg:Wu3}
    \end{subfigure}
%    \begin{subfigure}[b]{0.23\textwidth}
%    \centering
%        \begin{tikzpicture}[scale=0.9]
%  \node[anchor=south west,inner sep=0] at (0,0) {\includegraphics[width=\textwidth]{figures/ProtoData1}};
%  \node[anchor=south west,inner sep=0] at (0,-2.5) {\includegraphics[width=\textwidth]{figures/ProtoData1NullDist}};
%    \end{tikzpicture}
%        \caption{Poor colocalization is expected between AXR4-GFP and $\gamma$-Cop. (size: $312\times 312$)}
%        \label{fg:Wu4}
%    \end{subfigure}
%    \begin{subfigure}[b]{0.23\textwidth}
%    \centering
%        \begin{tikzpicture}[scale=0.9]
%  \node[anchor=south west,inner sep=0] at (0,0) {\includegraphics[width=\textwidth]{figures/ProtoData2}};
%  \node[anchor=south west,inner sep=0] at (0,-2.5) {\includegraphics[width=\textwidth]{figures/ProtoData2NullDist}};
%    \end{tikzpicture}
%        \caption{Good colocalization is expected between GFP-LTI6a and N-YFP-AUX1. (size: $310\times 310$)}
%        \label{fg:Wu5}
%    \end{subfigure}
%    \begin{subfigure}[b]{0.23\textwidth}
%    \centering
%        \begin{tikzpicture}[scale=0.9]
%  \node[anchor=south west,inner sep=0] at (0,0) {\includegraphics[width=\textwidth]{figures/ProtoData3}};
%  \node[anchor=south west,inner sep=0] at (0,-2.5) {\includegraphics[width=\textwidth]{figures/ProtoData3NullDist}};
%    \end{tikzpicture}
%        \caption{Good colocalization is expected bewteen AXR4-GFP and BIP. (size: $318\times 318$)}
%        \label{fg:Wu6}
%    \end{subfigure}
    \caption{Colocalization analysis result of $\tau^\ast_{\rm app}$ on benchmark real data examples from \cite{ICCSN}.}\label{fg:Wu}
\end{figure*}

\subsection{Real Data Examples}
\label{sc:real}

Finally, we applied our new method to  real biological datasets. The first example is a set of microscopic images (image size: $1024\times 1024$) of HeLa cells expressing the structural protein, Gag, of human immunodeficiency virus type 1 (HIV-1).  HIV-1 virus particles assemble at the plasma membrane and are composed of $\sim$2000 molecules of Gag \cite{freed2015,sundquist2012}.  We applied the same analysis procedures as for the previous sections. There are three conditions with corresponding images. In the first two conditions (Figure~\ref{fg:Real1} and Figure~\ref{fg:Real2}), HIV-1-Gag (green channel) was fused to cyan fluorescence protein (CFP) and MS2 protein (red channel) was fused to yellow fluorescent protein (YFP). When expressed in cells as the only viral factor, HIV-1-Gag primarily forms particles at the edge of cells; these particles are only occasionally internalized by the cell and then observed near the nucleus. In the first condition, MS2 protein was designed to remain in the nucleus (Figure~\ref{fg:Real1}), resulting in a negative control with low levels of colocalization between Gag-CFP and MS2-YFP.  In the second condition (Figure~\ref{fg:Real2}), Gag-CFP was expressed from an mRNA engineered to contain multiple copies of an RNA stem loop that binds MS2-YFP with high specificity \cite{becker2017}.  Therefore, we expected significantly higher colocalization levels between Gag-CFP and MS2-YFP in Figure~\ref{fg:Real2} as compared to those in Figure~\ref{fg:Real1}. We summarized $p$-values and the corresponding approximated null distributions in Figure~\ref{fg:realdata}. The results show colocalization was discovered in Figure~\ref{fg:Real2} if we rejected the null hypothesis when the $p$-value was smaller than 10\%. On the other hand, no significant colocalization was found as $p$-values in Figure~\ref{fg:Real1} were both larger than 70\%.  In the final condition (Figure~\ref{fg:Real3}), two constructs expressing synthetic Gags were fused to CFP and YFP, respectively. As Gag should self-assemble into multi-colored particles, we expected the highest levels of colocalization in this condition between Gag-CFP and Gag-YFP as compared to the two previous conditions. After applying our new method on these images, we obtained a very strong, significant level of colocalization, with $p$-values far less than 0.1\%.

We also applied our new method to another set of biological datasets.  These microscopic images (image size: $512\times512$) represent snapshots of a model used to elucidate signal responses during cellular wounding and the subsequent repair process.  Rho GTPases, including Rho and Cdc42, control an enormous variety of processes and play a role during {\it Xenopus} oocyte wound repair \cite{simon2013}; however, they do not overlap during the wound repair process and therefore resulted in low levels, $p$-values larger than 85\%, of detectible colocalization (Figure~\ref{fg:Bill1}). Calcium is an initially crude signal in wound repair, and PKC$\beta$ participates in Rho and Cdc42 activation and is also recruited to cell wounds \cite{vaughan2014}. Calcium defines a broad region within which PKC$\beta$ can be found, and therefore, some level of colocalization is expected, which was easily detected using our method (Figure~\ref{fg:Bill2}).  Finally, Rho GTPases including Rho and Cdc42, have also been implicated in cortical cytoskeleton repair, so the actin regulatory protein, cortactin, largely overlaps with Cdc42, for example, during the wound healing process.  The highest levels of colocalization were expected between Cdc42 and cortactin within this group of images, and this was measured by our method (Figure~\ref{fg:Bill3}). Once again, this work demonstrates our new method's robustness within complex, biological contexts.

For the microscopic images of both biological datasets, we also applied Pearson correlation coefficient $r$ and Manders' split coefficients $(M_1,M_2)$. For $M_1$ and $M_2$, the thresholds $\alpha_X$ and $\alpha_Y$ are still determined by Otsu's method. In $\tau^\ast_{\rm app}$, the lower bound of thresholds scanned was chosen as maximum of Otsu's threshold and median value. To obtain a $p$-value, the microscopic images were permuted block-wise as described in Section~\ref{sc:fast}. In these experiments, 1000 permutations were carried out and the block size was 32, the square root of the size of the image. The value of colocalization measures and corresponding $p$-values calculated by the permutation test are summarized in Table~\ref{tb:realdata}. The results in Table~\ref{tb:realdata} suggest that our new statistics $\tau^\ast_{\rm app}$ is able to control false discovery far better than Pearson correlation coefficient and Manders' split coefficients. Moreover, the value of our $\tau^\ast_{\rm app}$ and corresponding $p$-value can reflect the level of colocalization more precisely. It is worth noting that the size of the newly proposed index $\tau_{\rm app}^\ast$ can also be affected by the area of colocalized region. In other words, $\tau_{\rm app}^\ast$ is relatively small when the colocalization happens in a small region. For example, $\tau_{\rm app}^\ast$ is relatively small ($p$-value is relatively large) in Figure~\ref{fg:Real2} when colocalization only concentrates at the edge of cell.

\begin{figure}[!htbp]
    \centering
    \begin{subfigure}[b]{0.5\textwidth}
    \centering
        \begin{tikzpicture}[scale=1]
        \draw[line width=2pt,black!80!white] (-0.25,4.35) rectangle (8.6,-2.35);
  \node[anchor=south west,inner sep=0] at (0,-2.2) {\includegraphics[width=0.45\textwidth]{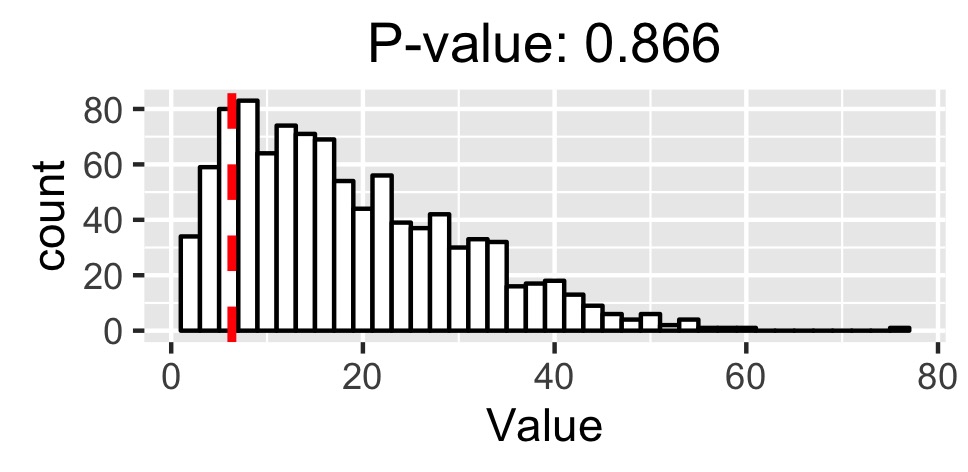}};
  \node[anchor=south west,inner sep=0] at (0,0) {\includegraphics[width=0.45\textwidth]{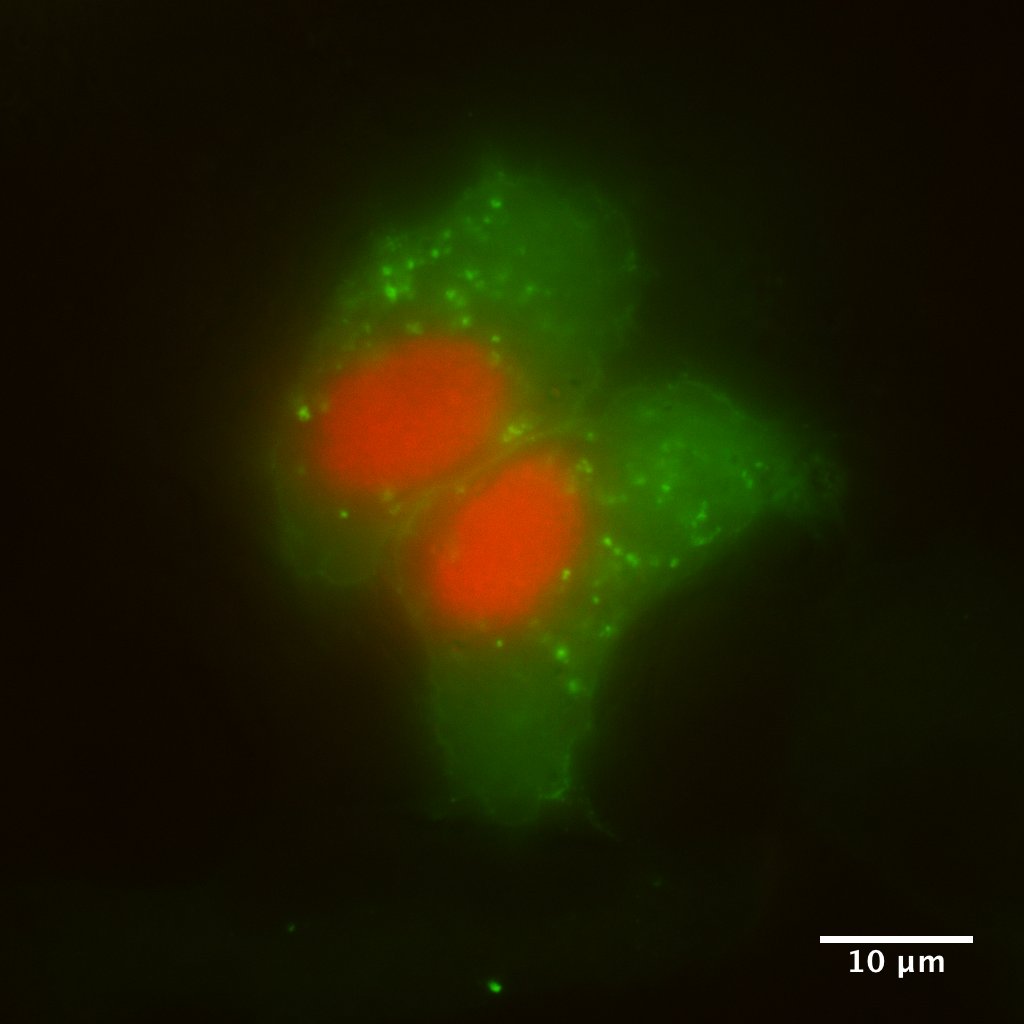}};
   \node[anchor=south west,inner sep=0] at (4.25,-2.2) {\includegraphics[width=0.45\textwidth]{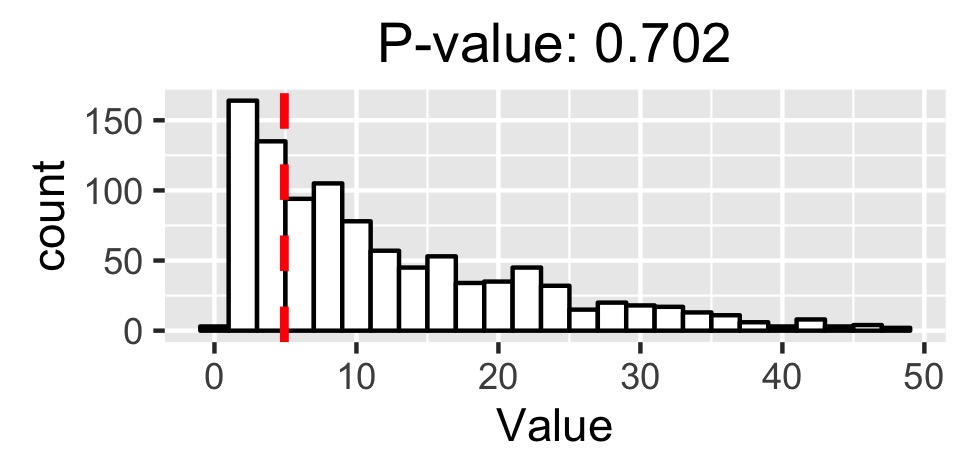}}; 
  \node[anchor=south west,inner sep=0] at (4.25,0) {\includegraphics[width=0.45\textwidth]{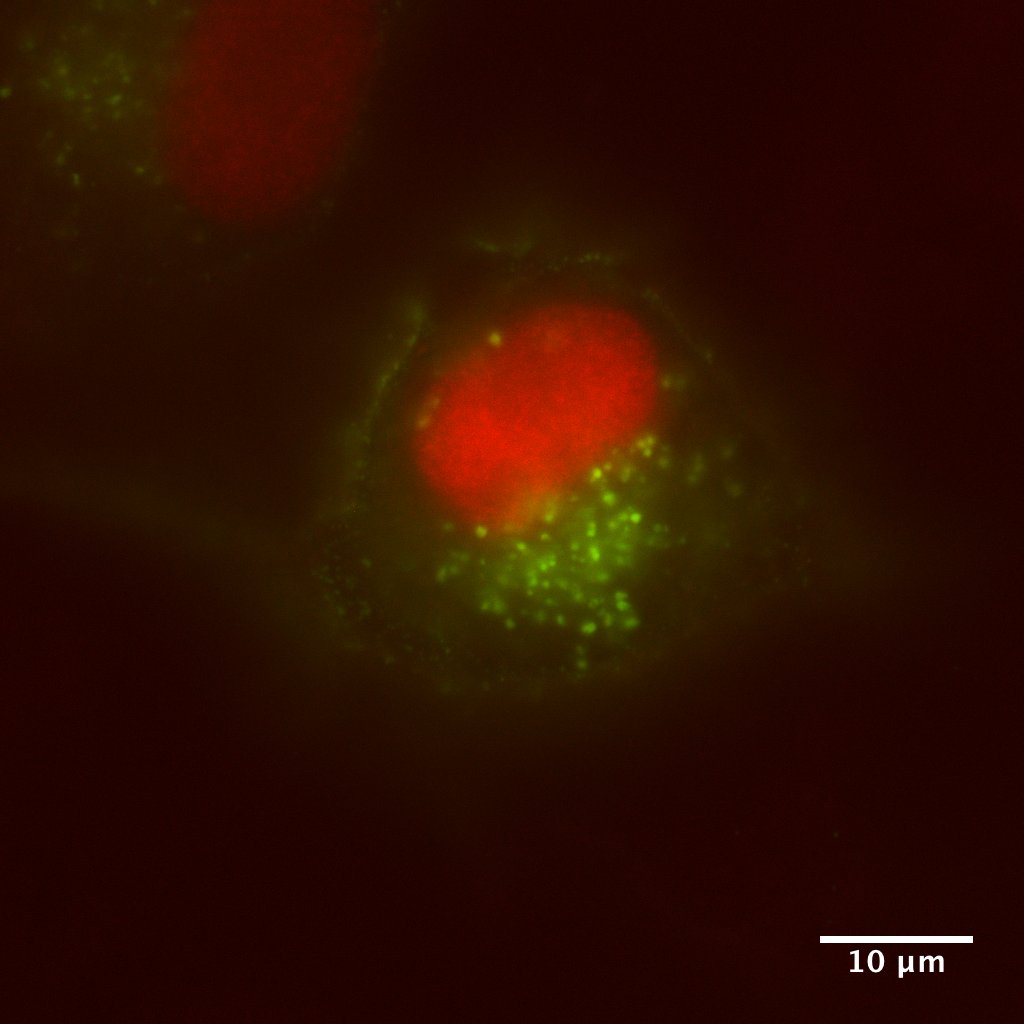}};
    \end{tikzpicture}
        \caption{Poor colocalization examples (image size: $1024\times 1024$): Low levels of colocalization between MS2-YFP and Gag-CFP are expected.}
        \label{fg:Real1}
    \end{subfigure}
    \begin{subfigure}[b]{0.5\textwidth}
    \centering
        \begin{tikzpicture}[scale=1]
        \draw[line width=2pt,black!80!white] (-0.25,4.35) rectangle (8.6,-2.35);
  \node[anchor=south west,inner sep=0] at (0,-2.2) {\includegraphics[width=0.45\textwidth]{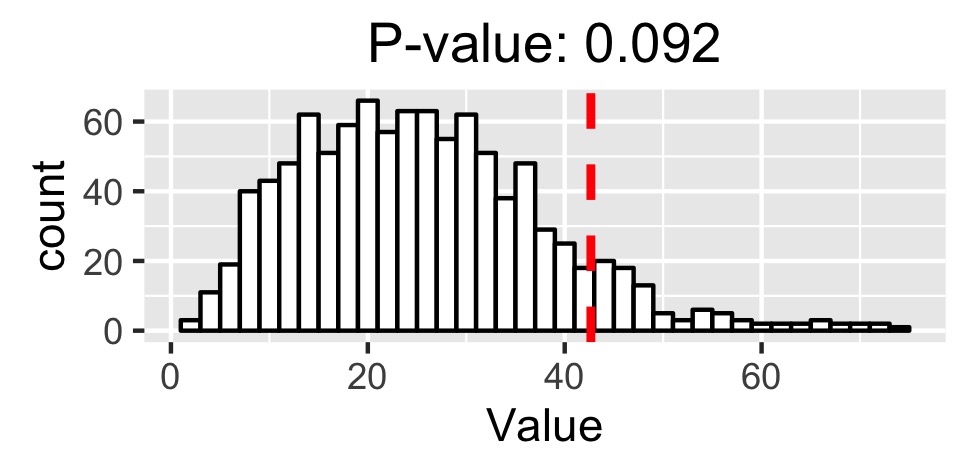}};
  \node[anchor=south west,inner sep=0] at (0,0) {\includegraphics[width=0.45\textwidth]{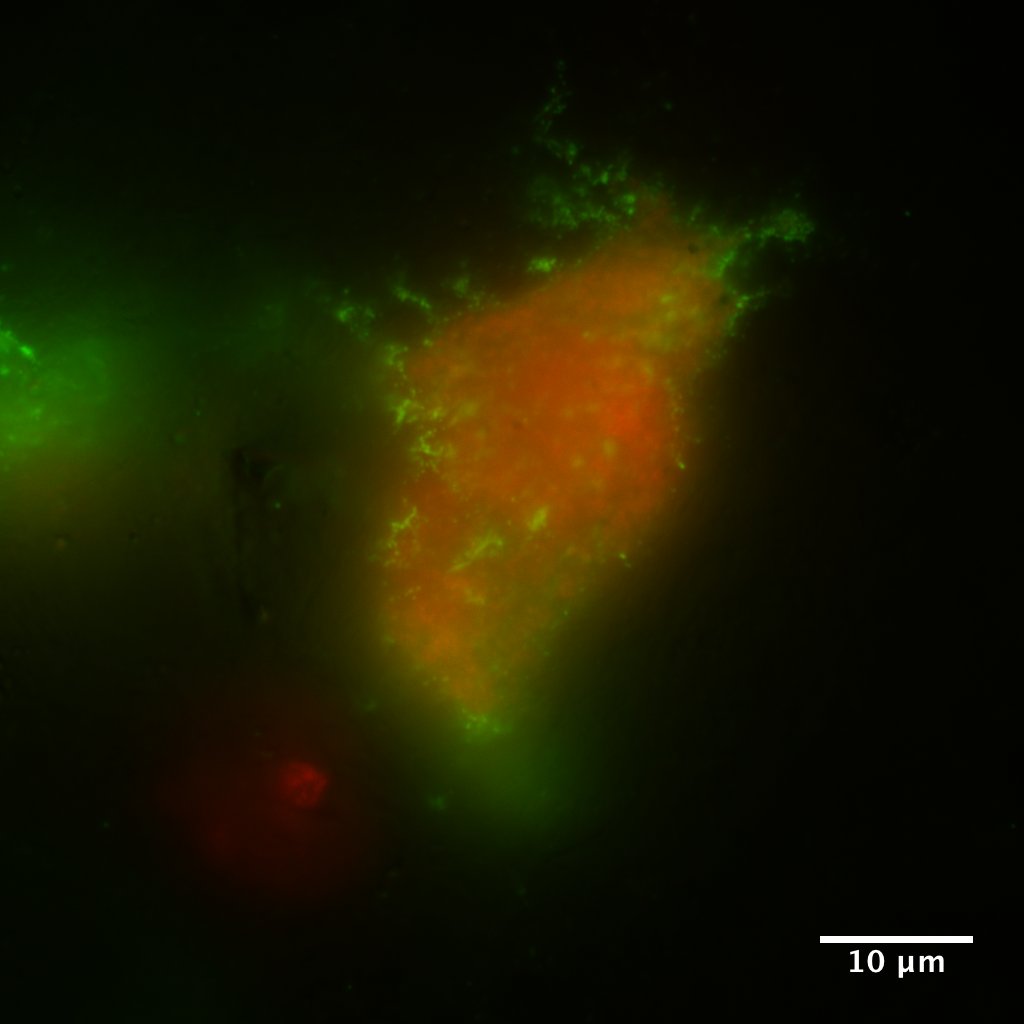}};
   \node[anchor=south west,inner sep=0] at (4.25,-2.2) {\includegraphics[width=0.45\textwidth]{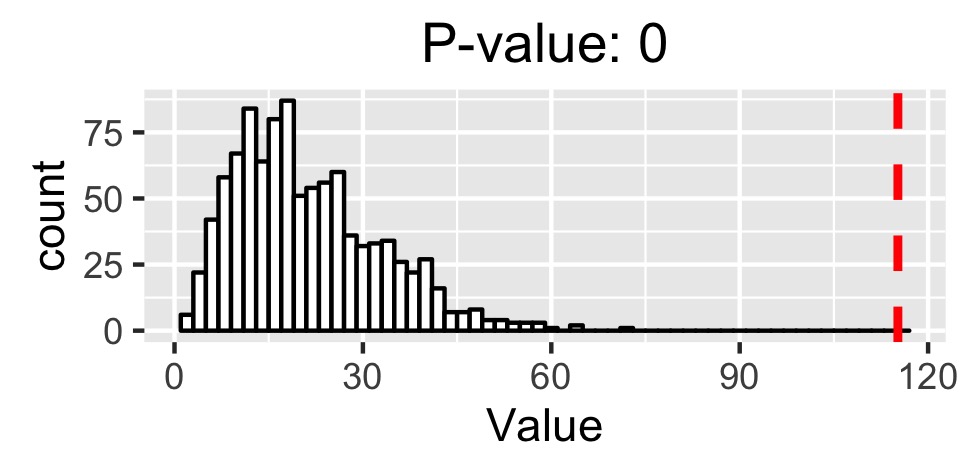}};  
  \node[anchor=south west,inner sep=0] at (4.25,0) {\includegraphics[width=0.45\textwidth]{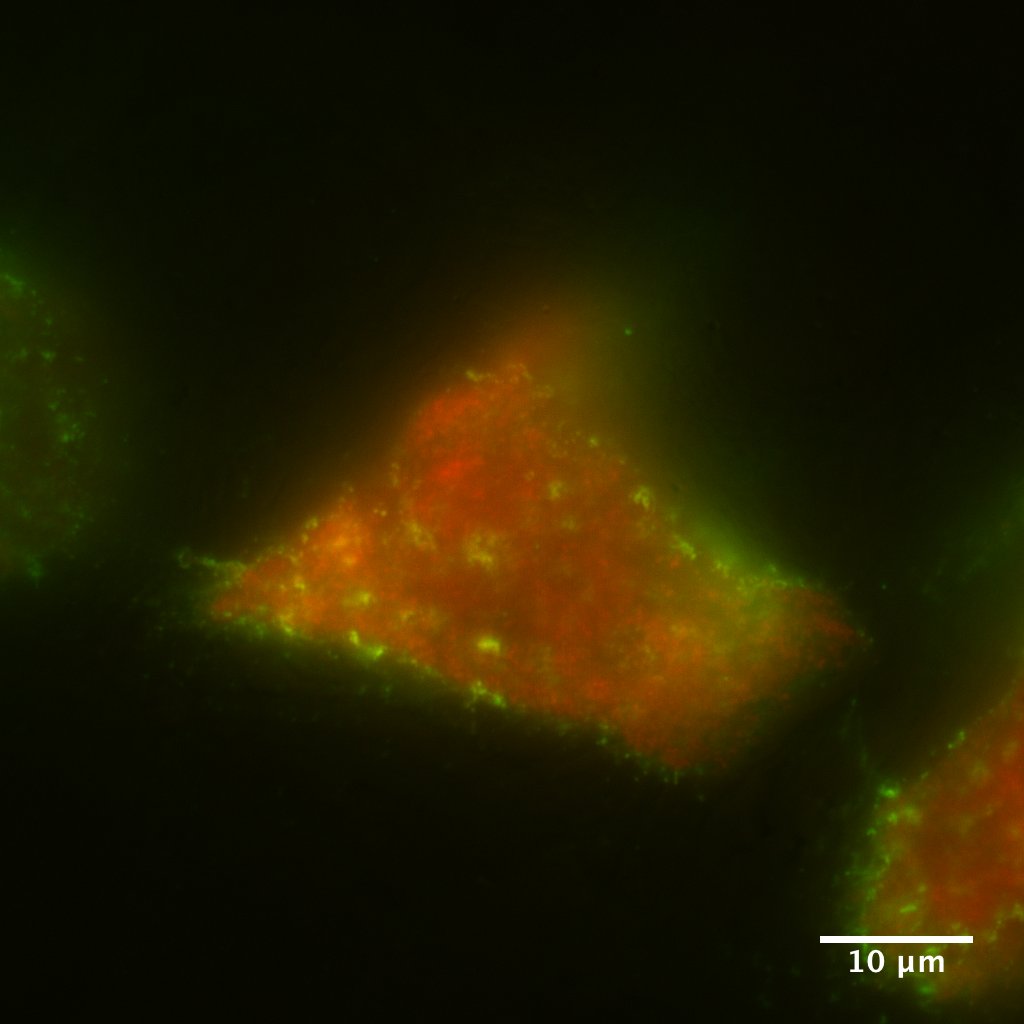}};
    \end{tikzpicture}
        \caption{Good colocalization examples (image size: $1024\times 1024$): High level of colocalization between Gag-CFP and MS2-YFP are expected.}
        \label{fg:Real2}
    \end{subfigure}
    \begin{subfigure}[b]{0.5\textwidth}
    \centering
        \begin{tikzpicture}[scale=1]
        \draw[line width=2pt,black!80!white] (-0.25,4.35) rectangle (8.6,-2.35);
  \node[anchor=south west,inner sep=0] at (0,-2.2) {\includegraphics[width=0.45\textwidth]{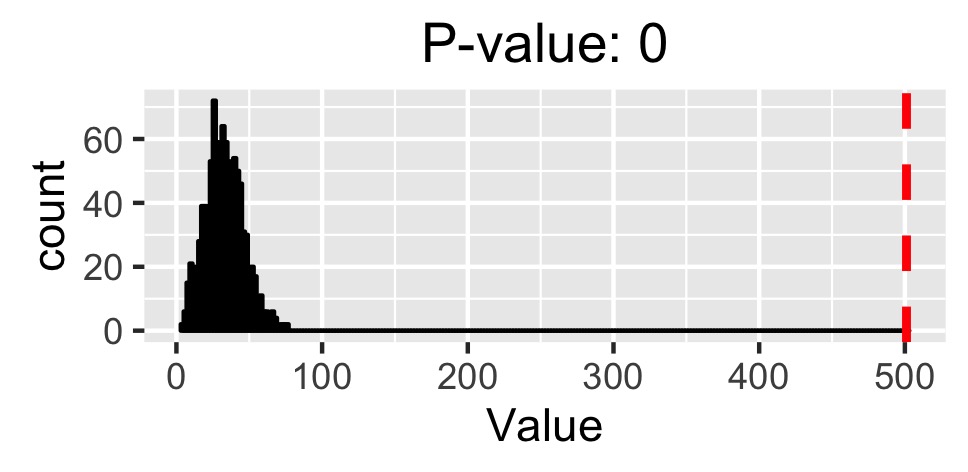}};
  \node[anchor=south west,inner sep=0] at (0,0) {\includegraphics[width=0.45\textwidth]{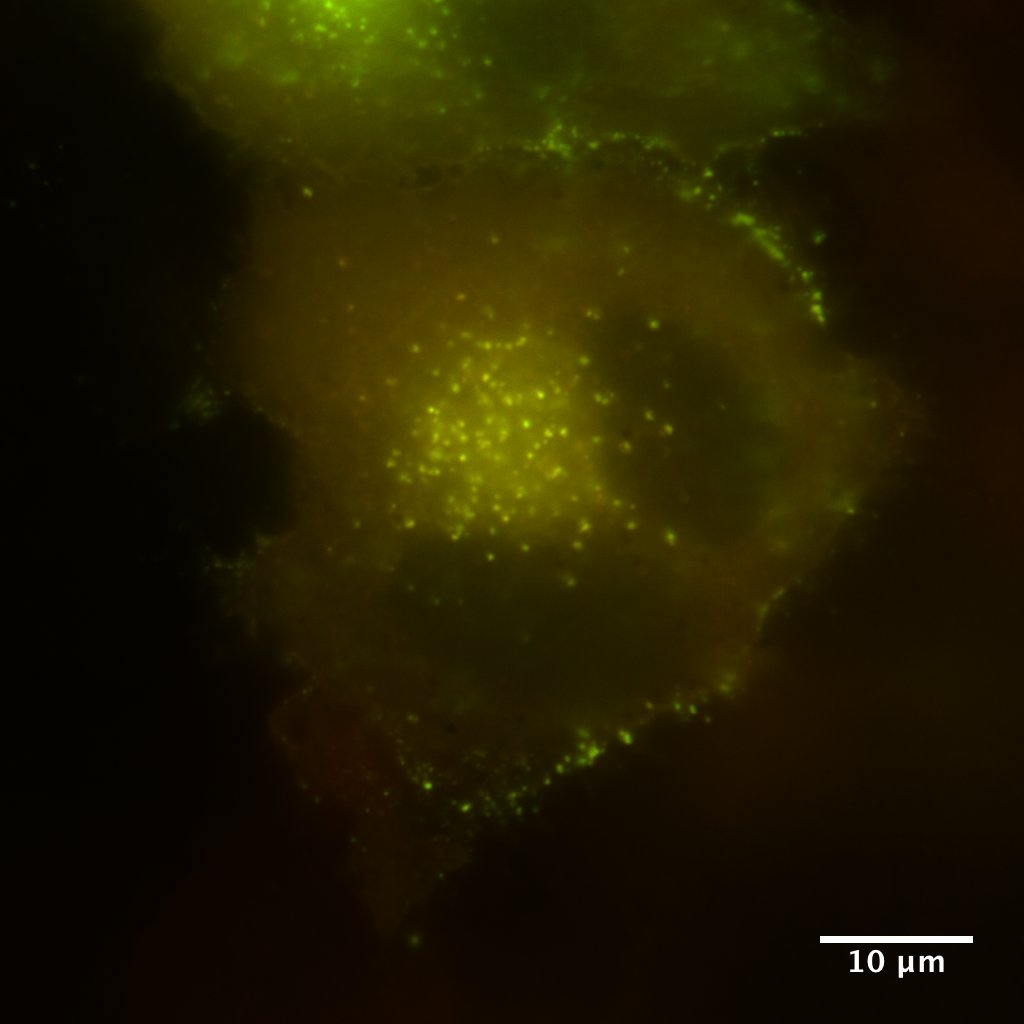}};
  \node[anchor=south west,inner sep=0] at (4.25,-2.2) {\includegraphics[width=0.45\textwidth]{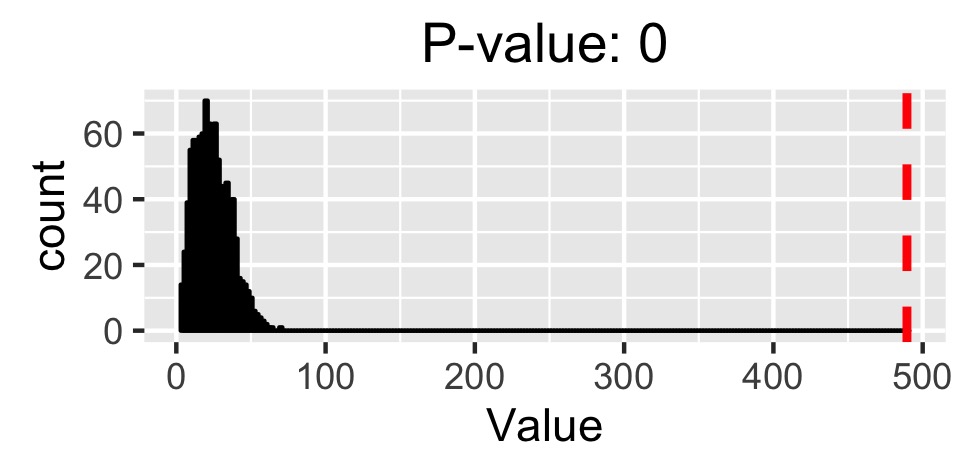}};  
  \node[anchor=south west,inner sep=0] at (4.25,0) {\includegraphics[width=0.45\textwidth]{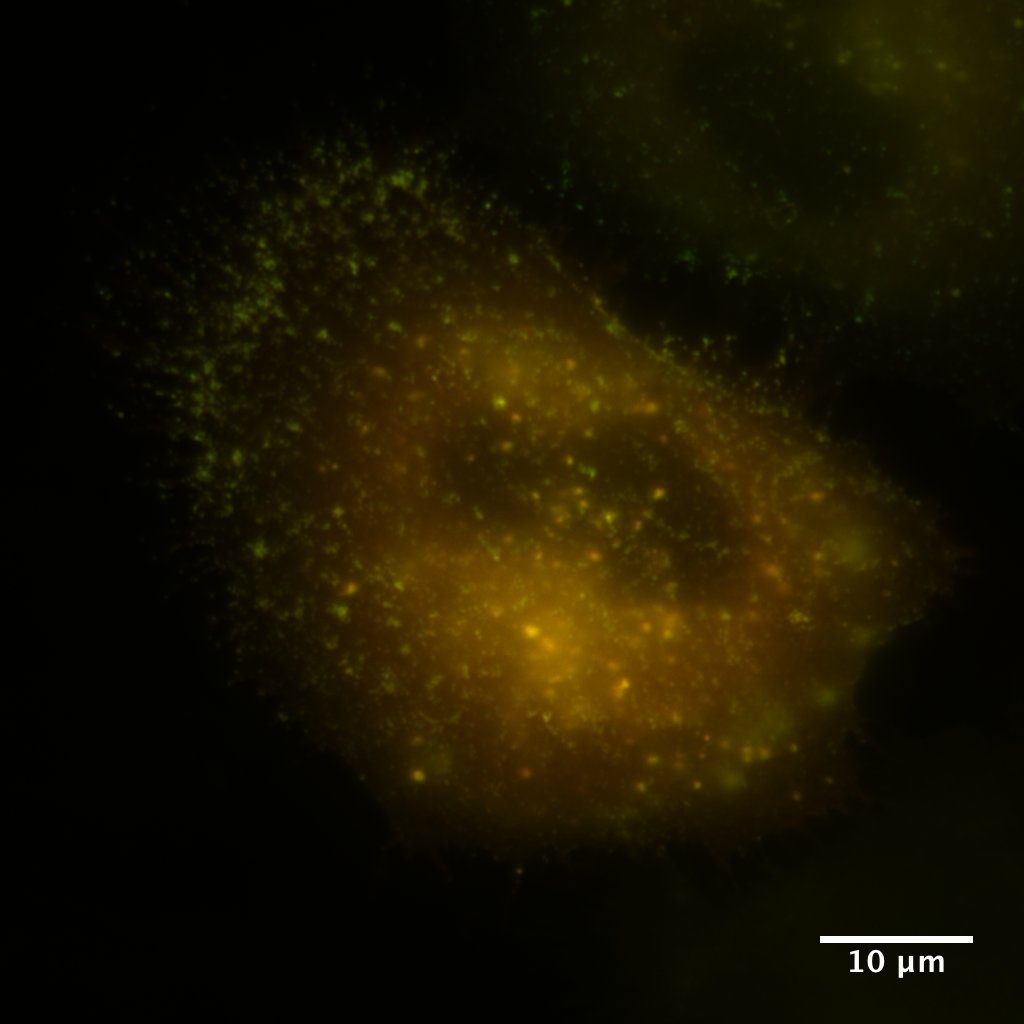}}; 
    \end{tikzpicture}
        \caption{Strong colocalization examples (image size: $1024\times 1024$): The highest level of colocalization between Gag-CFP and Gag-YFP are expected.}
        \label{fg:Real3}
    \end{subfigure}
    \caption{$p$-value and null distribution obtained by our new method on real data examples of Hella cells.}\label{fg:realdata}
\end{figure}

\begingroup
\renewcommand{\arraystretch}{1.5}
\begin{table*}[!htbp]
\centering
\begin{tabular}{c c c c c c c c c c c c c c c}
\hline
\hline
& & & & \multicolumn{2}{c}{Pearson $r$} & &  \multicolumn{2}{c}{Manders $M_1$}& &\multicolumn{2}{c}{Manders $M_2$}& &\multicolumn{2}{c}{New measure $\tau^\ast_{\rm app}$} \\
 \cline{5-6}  \cline{8-9}  \cline{11-12} \cline{14-15}
& & & & $r$ & $p$-value & & $M_1$ & $p$-value & & $M_2$ & $p$-value & & $\tau^\ast_{\rm app}$ & $p$-value \\
\hline
\multirow{2}{*}{Figure~\ref{fg:Real1} (Poor colocalization)}& & left         & &  $ 0.599 $ & $<0.1\%$ & &  $0.425$ & $<0.1\%$ & &  $0.122$ & $<0.1\%$ & &  $6.328$ & $86.6\%$ \\ 
& & right                                                                                                & &  $ 0.581 $ & $<0.1\%$ & &  $0.276$ & $<0.1\%$ & &  $0.102$ & $ <0.1\% $ & &  $4.922$ & $70.2\%$ \\ 
\hline
\multirow{2}{*}{Figure~\ref{fg:Real2} (Good colocalization)} & & left   & &  $ 0.843 $ & $<0.1\% $ & &  $0.629$ & $<0.1\%$ & &  $0.308$ & $<0.1\%$ & &  $42.659$ & $9.2\%$ \\ 
& & right                                                                                                & &  $ 0.909 $ & $<0.1\%$ & &  $0.629$ & $<0.1\%$  & &  $0.355$ & $<0.1\%$ & &  $115.191$ & $<0.1\%$ \\ 
\hline
\multirow{2}{*}{Figure~\ref{fg:Real3} (Strong colocalization)} &  & left & &  $ 0.972 $ & $<0.1\%$ & &  $0.552$ & $<0.1\%$ & &  $0.658$ & $<0.1\%$ & &  $501.111$ & $<0.1\%$ \\ 
& & right                                                                                                & &  $ 0.983 $ & $<0.1\%$ & &  $0.609$ & $<0.1\%$ & &  $0.539$ & $<0.1\%$ & &  $489.459$ & $<0.1\%$ \\ 
\hline
\multicolumn{3}{c}{Figure~\ref{fg:Bill1} (Poor colocalization)} & &  $ 0.486 $ & $<0.1\%$ & &  $0.206$ & $<0.1\%$ & &  $0.123$ & $<0.1\%$ & &  $1.691$ & $85.1\%$\\
\multicolumn{3}{c}{Figure~\ref{fg:Bill2} (Good colocalization)} & &  $ 0.850 $ & $<0.1\%$ & &  $0.285$ & $<0.1\%$ & &  $0.178$ & $<0.1\%$ & &  $39.175$ & $<0.1\%$\\
\multicolumn{3}{c}{Figure~\ref{fg:Bill3} (Good colocalization)} & &  $ 0.291 $ & $<0.1\%$ & &  $0.122$ & $<0.1\%$ & &  $0.104$ & $<0.1\%$ & &  $15.798$ & $<0.1\%$\\
\hline
\hline
\end{tabular}
\caption{The colocalization measure values and corresponding $p$-value obtained by Pearson's correlation coefficient $r$, Manders' split coefficients $M_1$, $M_2$ and our new method $\tau^\ast_{\rm app}$ on microscopic images in Figure~\ref{fg:realdata} and Figure~\ref{fg:Bill}.}\label{tb:realdata}
\end{table*}
\endgroup

\begin{figure*}
    \centering
    \begin{subfigure}[b]{0.23\textwidth}
    \centering
        \begin{tikzpicture}[scale=0.9]
  \node[anchor=south west,inner sep=0] at (0,0) {\includegraphics[width=\textwidth]{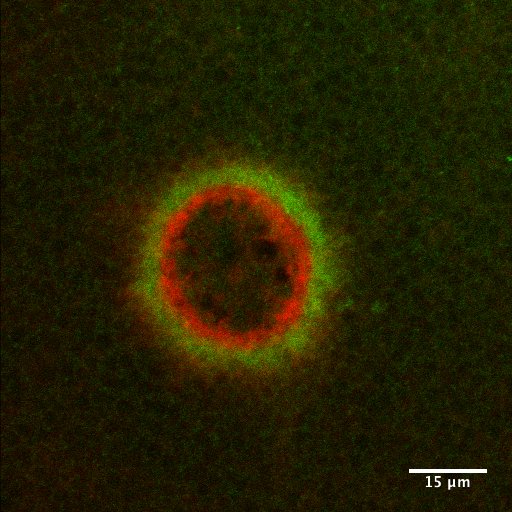}};
  \node[anchor=south west,inner sep=0] at (0,-2.5) {\includegraphics[width=\textwidth]{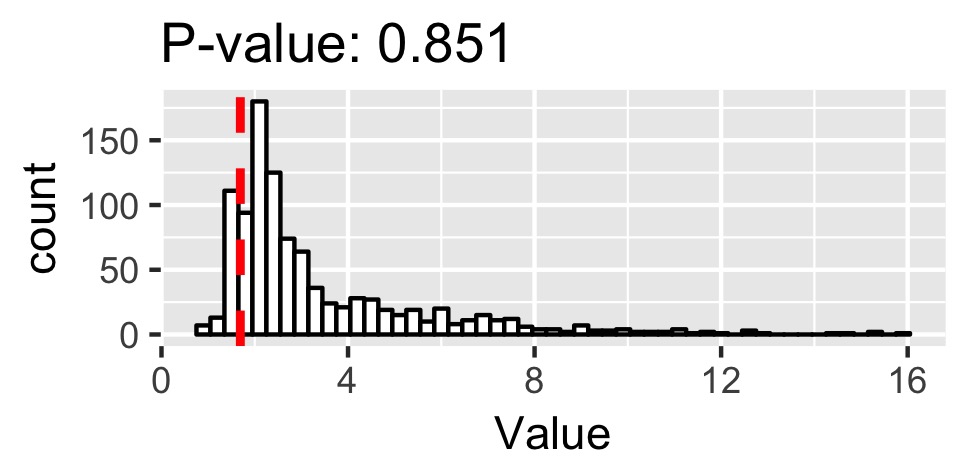}};
    \end{tikzpicture}
        \caption{Low levels of colocalization between Rho and Cdc42 are expected. (size: $512\times 512$)}
        \label{fg:Bill1}
    \end{subfigure}
    \hspace{0.05\textwidth}
    \begin{subfigure}[b]{0.23\textwidth}
    \centering
        \begin{tikzpicture}[scale=0.9]
  \node[anchor=south west,inner sep=0] at (0,0) {\includegraphics[width=\textwidth]{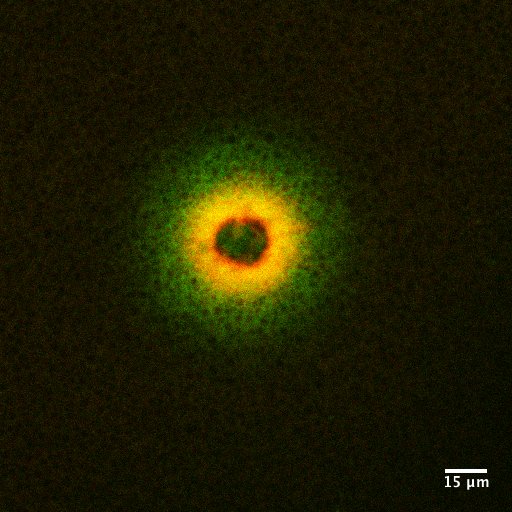}};
  \node[anchor=south west,inner sep=0] at (0,-2.5) {\includegraphics[width=\textwidth]{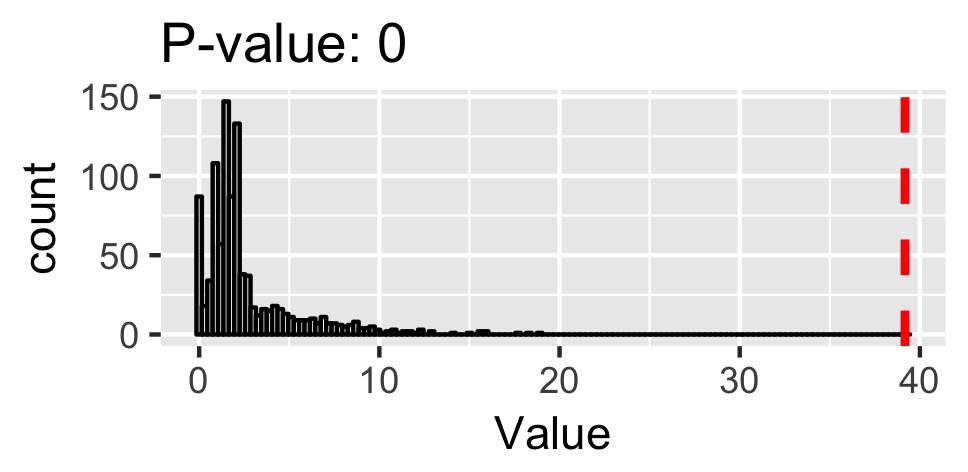}};
    \end{tikzpicture}
        \caption{High levels of colocalization between PKC$\beta$ and calcium are expected. (size: $512\times 512$)}
        \label{fg:Bill2}
    \end{subfigure}
    \hspace{0.05\textwidth}
    \begin{subfigure}[b]{0.23\textwidth}
    \centering
        \begin{tikzpicture}[scale=0.9]
  \node[anchor=south west,inner sep=0] at (0,0) {\includegraphics[width=\textwidth]{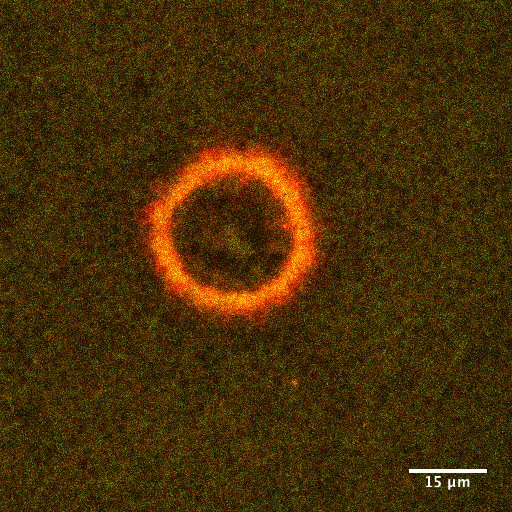}};
  \node[anchor=south west,inner sep=0] at (0,-2.5) {\includegraphics[width=\textwidth]{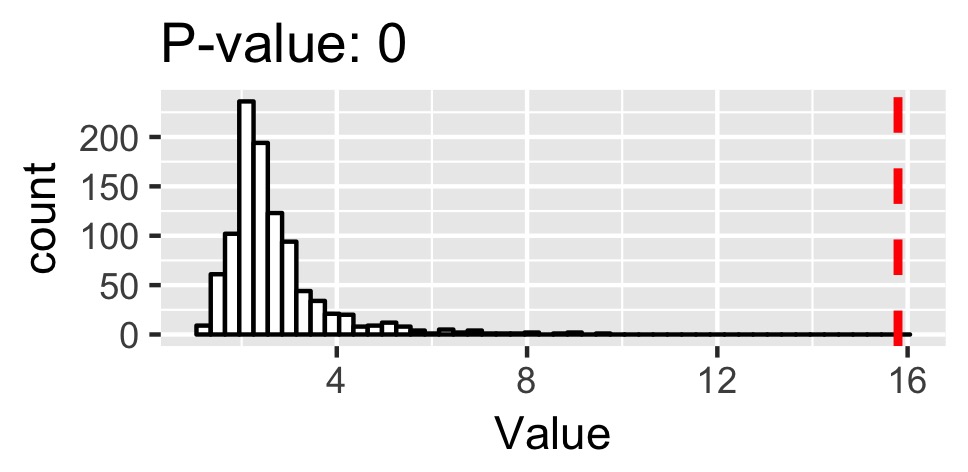}};
    \end{tikzpicture}
        \caption{High levels of colocalization between Cdc42 and cortactin are expected. (size: $512\times 512$)}
        \label{fg:Bill3}
    \end{subfigure}
    \caption{$p$-value and null distribution obtained by our new method on real data examples of wound cell.}\label{fg:Bill}
\end{figure*}

\section{Concluding Remarks}
In this paper, we propose a new robust measure of colocalization. Due to the intrinsic, nonparametric characteristic of Kendall's tau correlation coefficient, the new colocalization measure captures a wider range of associations between two channels than most existing parametric quantitative measures, such as Pearson's correlation coefficient and Manders' split coefficients. Given the vast complexites in bioimage data and variable associations between two biological probes beyond simple linear correlation or co-occurrence, our new nonparametric measure provides a more accurate reflection on the given association. Scanning at different signal levels allows our new measure to discover potential associations between two probes automatically without knowledge of thresholds for background. Under statistical hypothesis testing framework, if we assume mild regular conditions for intensity distributions, the test based on our new colocalization score is able to achieve statistical optimality. 

We also developed a user-friendly, fast algorithm for our new colocalization measure so that the colocalization score can be translated into statistical significance efficiently. To overcome the computational hurdle of scanning, we proposed an approximation of our new colocalization measure to accelerate computation. In doing so, the approximated colocalization score can be calculated much more efficiently. Furthermore, we adopted a block-wise permutation test as in \cite{costes2004automatic} to evaluate the calculated $p$-value. Putting this all into a single algorithm, users are able to get a $p$-value with a single `click'. Results from several experiments using benchmark and biological data converge to the conclusion that our new algorithm remains highly efficient.

The algorithm is readily available in an R package, RKColocal, as described previously.  This code is also currently being adapted for incorporation into ImageJ, a popular open-source bioimage analysis software package \citep[see, e.g.,][]{arena2016quantitating}. This tool and its continued development will also help bridge statistics and bioimaging by providing such improved algorithms and methods that facilitate productive collaborations between fields. 

Interdisciplinary, collaborative research can lead to more innovations and discoveries. When colocalization analyses are cast as statistical hypothesis testing problems, as shown in Section~\ref{sc:opt}, we tackled a bioimage processing problem with statistical techniques without losing perspectives from both communities. The statistical hypothesis testing framework not only helped us develop an efficient approach to detect interesting associations between probes, but also made sure that true associations between channels were always reported and false discoveries kept under control. Through a statistical lens, our new nonparametric statistical approach is ultimately trustworthy and precise. We believe the same application of statistics can also be extended to other bioimage processing techniques, including deconvolution, spectral unmixing, lifetime analyses, and more. We anticipate more collaborative benefits at the intersection of bioimage processing and statistics in the future.

\section*{Acknowledgment}

{  The authors would like to thank Nathan Sherer and Jordan Becker for sharing the microscopy image data sets in Figure~\ref{fg:realdata}, and William Bement for sharing the the microscopy image data sets in Figure~\ref{fg:Bill}. }

% or


\appendix[Proof of Theorem~\ref{thm:alter}]
\label{sc:pf}

The proof is somewhat lengthy, and we break it into several steps.

\paragraph{Size of $q_\alpha$} We first show that
$$
q_\alpha=O_p(\sqrt{\log\log n}).
$$
Recall that $q_\alpha$ is the upper $\alpha$ quantile of $\tau^\ast$ under $H_0$. It then suffices to show that there exists some universal constant $C$ such that
\begin{equation}
\label{eq:null}
\PP\left(\tau^\ast\ge C\sqrt{\log\log n}\right)\to 0.
\end{equation}

Observe that $\tau^\ast$ does not depend on the marginal distribution under $H_0$, we can assume without loss of generality that $F(x,y)=1_{(0\le x\le 1,0\le y\le 1)}$. Let $G_n^X(t)$, $G_n^Y(t)$, and $G_n(t)$ be the empirical distribution functions of $X$, $Y$, and $(X,Y)$, respectively, that is
$$
G_n^X(t)={1\over n}\sum_{i=1}^n I_{(X_i\ge t)},\qquad  G_n^Y(t)={1\over n}\sum_{i=1}^n I_{(Y_i\ge t)},
$$
and
$$
G_n(t,s)={1\over n}\sum_{i=1}^n I_{(X_i\ge t,Y_i\ge s)}.
$$
Write
$$
A_\alpha^X:=\left\{\sup_{1/2\le t< 1}{|nG_n^X(t)-n(1-t)|\over \sqrt{n(1-t)}}\le \alpha \sqrt{\log\log n}\right\},
$$
$$
A_\alpha^Y:=\left\{\sup_{1/2\le t< 1}{|nG_n^Y(t)-n(1-t)|\over \sqrt{n(1-t)}}\le \alpha \sqrt{\log\log n}\right\},
$$
and
$$
A_\alpha:=\left\{\sup_{ \substack{{\log\log n/n}\le\\ L(t,s)\le 1/2}}{|nG_n(t,s)-nL(t,s)|\over \sqrt{nL(t,s)}}\le \alpha \sqrt{\log\log n}\right\}.
$$
Hereafter, we refer $L(t,s)=(1-t)(1-s)$.
It is well known that there exists $\alpha_0>2$ such that
$$
\PP\left(A_{\alpha_0}^X\cap A_{\alpha_0}^Y\cap A_{\alpha_0}\right)\to 1,\qquad n\to \infty.
$$
See, e.g., \cite{einmahl1996}. Hence, it is sufficient to get (\ref{eq:null}) conditioned on $A_{\alpha_0}^X\cap A_{\alpha_0}^Y\cap A_{\alpha_0}$.

Recall that
$$
\tau^\ast=\max_{j,k\ge \lfloor n/2\rfloor} \tilde{\tau}(X_{(j)},Y_{(k)}),
$$
where
\begin{equation}
\label{eq:tildetau}
\tilde{\tau}(t_X,t_Y)=\tau(t_X,t_Y)\cdot \sqrt{9n_{t_X,t_Y}(n_{t_X,t_Y}-1)\over 2(2n_{t_X,t_Y}+5)}.
\end{equation}
Write
$$
\Ical_1=\left\{(j,k)\in [n]^2: j,k\ge \lfloor n/2\rfloor, nL\left({j\over n},{k\over n}\right)\le \alpha^2_0\log\log n\right\},
$$
$$
\Ical_2=\left\{(j,k)\in [n]^2: j,k\ge \lfloor n/2\rfloor, nL\left({j\over n},{k\over n}\right)> \log^2n\right\},
$$
and
\begin{equation}
\begin{split}
\Ical_3=\{& (j,k)\in [n]^2: j,k\ge \lfloor n/2\rfloor,\\
& \alpha^2_0\log\log n\le nL\left({j\over n},{k\over n}\right)\le \log^2n\}.
\end{split}
\end{equation}
It is clear that $\tau^\ast=\max\{\tau_1^\ast,\tau_2^\ast,\tau_3^\ast\}$ where
$$\tau_j^\ast=\max_{(j,k)\in \Ical_j} \tilde{\tau}(X_{(j)},Y_{(k)}).$$
It therefore suffices to upper bound $\tau_j^\ast$ separately.

We begin with $\tau_1^\ast$. Under the event $A_{\alpha_0}^X\cap A_{\alpha_0}^Y\cap A_{\alpha_0}$, we have
$$
(1-X_{(j)})(1-Y_{(k)})\le {2\alpha^2_0\log\log n \over n},
$$
when $n(1-j/n)(1-k/n)\le \alpha^2_0\log\log n$. As shown by \cite{einmahl1996},
\begin{align*}
&\PP\left(\sup_{L(t,s)\le 2\alpha^2_0\log\log n/n}nG_n(t,s)\ge m\right)\\
\le & c_1{(2c_2\alpha^2_0\log\log n)^m \over m!}\log n.
\end{align*}
Hereafter, we shall use $c$ to denote a generic positive constant that may take different values at each appearance. Taking $m=2c_2\alpha^2_0e^2\log\log n$, we can ensure
$$
\PP\left(\sup_{L(t,s)\le 2\alpha^2_0\log\log n/n}nG_n(t,s)\ge 2c_2\alpha^2_0e^2\log\log n\right)\to 0.
$$
This suggests
$$
\PP\left(\max_{(j,k)\in\Ical_1} n_{X_{(j)},Y_{(k)}}\ge 2c_2\alpha^2_0e^2\log\log n\right)\to 0.
$$
By definition of $\tilde{\tau}(t_X,t_Y)$, 
$$
\tilde{\tau}(t_X,t_Y)\le \sqrt{ 9n_{t_X,t_Y}(n_{t_X,t_Y}-1) \over 2(2n_{t_X,t_Y}+5) }.
$$
This immediately suggests
$$
\PP\left(\tau_1^\ast>{3\over 2}\sqrt{2c_2\alpha^2_0e^2\log\log n}\right)\to 0.
$$

Next, we consider $\tau_2^\ast$.
Let 
$$
\Ncal_n=\left\{(j,k):j,k\in \Scal_n \quad {\rm and}\quad nL\left({j\over n},{k\over n}\right)> \log^2n \right\}
$$
where $\Scal_n$ is
$$
\left\{s:s=\left\lfloor n-\left(1+{\sqrt{\log\log n}\over \log n}\right)^j\right\rfloor, j\in\mathbb{N}_{+}, s\ge  \lfloor n/2\rfloor \right\},
$$
and 
$$
\tau_{\Ncal_n}^\ast:=\max_{(j,k)\in \Ncal_n} \tilde{\tau}(X_{(j)},Y_{(k)}).
$$
Our strategy is to first show the difference between $\tau_{\Ncal_n}^\ast$ and $\tau_2^\ast$ is negligible and then bound $\tau_{\Ncal_n}^\ast$.

To bound $\tau_2^\ast-\tau_{\Ncal_n}^\ast$, we consider the following projection $\pi:\NN\to \NN$, which maps integer $i$ to the largest integer in $\Scal_n$ that is smaller than $i$. Conditioned on $A_{\alpha_0}^X\cap A_{\alpha_0}^Y\cap A_{\alpha_0}$,
\begin{align*}
&n_{X_{(j)},Y_{(k)}}\\
\le& nL\left(X_{(j)},Y_{(k)}\right)\left(1+{\alpha_0\sqrt{\log\log n}\over\sqrt{n(1-X_{(j)})(1-Y_{(k)})}}\right)\\
\le& nL\left({j\over n},{k\over n}\right)\left(1+{2\alpha_0\sqrt{\log\log n}\over\log n}\right)^3
\end{align*}
Therefore,
\begin{align*}
&n_{X_{(\pi(j))},Y_{(\pi(k))}}\\
&\ge nL\left(X_{(\pi(j))},Y_{(\pi(k))}\right)\left(1-{\alpha_0\sqrt{\log\log n}\over\sqrt{n(1-X_{(\pi(j))})(1-Y_{(\pi(k))})}}\right)\\
&\ge nL\left({\pi(j)\over n},{\pi(k)\over n}\right)\left(1-{2\alpha_0\sqrt{\log\log n}\over\log n}\right)^3\\
&\ge nL\left({j\over n},{k\over n}\right)\left(1-{2\alpha_0\sqrt{\log\log n}\over\log n}\right)^3\left(1+{\sqrt{\log\log n}\over\log n}\right)^{-2}
\end{align*}
This implies
\begin{equation}
\label{eq:numineq1}
n_{X_{(j)},Y_{(k)}}\le n_{X_{(\pi(j))},Y_{(\pi(k))}}\left(1+{20\alpha_0\sqrt{\log\log n}\over\log n}\right),
\end{equation}
for sufficiently large $n$. We then appeal to the following technical result.

\begin{lemma}
\label{lm:diffconcen}
Let $X$ and $Y$ be two independent uniform random variables. For two fixed pairs $(t_{X},t_{Y})$ and $(t'_{X},t'_{Y})$, denote by
$$
B=\{n_{(\min(t_{X},t'_{X}),\min(t_{Y},t'_{Y}))}\le (1+\epsilon)n_{(\max(t_{X},t'_{X}),\max(t'_{Y},t'_{Y}))}\}.
$$
Then
\begin{equation}
\label{eq:diffconcen}
\PP\left(|\tilde{\tau}(t_X,t_Y)-\tilde{\tau}(t'_X,t'_Y)|>r| B \right)\le 4\exp\left(-{r^2\over 72\epsilon^2+18\epsilon}\right),
\end{equation}
where $\tilde{\tau}(t_X,t_Y)$ is defined in (\ref{eq:tildetau}). In particular, if $\epsilon\le 1/12$, we have
$$
\PP\left(|\tilde{\tau}(t_X,t_Y)-\tilde{\tau}(t'_X,t'_Y)|>r| B \right)\le 4\exp\left(-{r^2\over 24\epsilon}\right).
$$
\end{lemma}

Lemma~\ref{lm:diffconcen} immediately suggests that 
\begin{align*}
&\PP\left(|\tilde{\tau}(X_{(j)},Y_{(k)})-\tilde{\tau}(X_{(\pi(j))},Y_{(\pi(k))})|>r\right)\\
%=&\PP\left(|\tilde{\tau}(X_{(j)},Y_{(k)})-\tilde{\tau}(X_{(\pi(j))},Y_{(\pi(k))})|>r\middle| n_{X_{(j)},Y_{(k)}}\le n_{X_{(\pi(j))},Y_{(\pi(k))}}\left(1+{20\alpha_0\sqrt{\log\log n}\over\log n}\right) \right)\\
%\le& \sup_{j,k}\PP\left(|\tilde{\tau}(X_{(j)},Y_{(k)})-\tilde{\tau}(X_{(\pi(j))},Y_{(\pi(k))})|>r\middle| n_{X_{(j)},Y_{(k)}}\le n_{X_{(\pi(j))},Y_{(\pi(k))}}\left(1+{20\alpha_0\sqrt{\log\log n}\over\log n}\right) \right)\\
\le& 4\exp\left(-{r^2\log n\over 480\alpha_0\sqrt{\log\log n}}\right).
\end{align*}
Because there are at most $n^2$ pairs $(j,k)$ and $\log^4 n$ pairs $(\pi(j),\pi(k))$, an application of union bound yields
$$
\PP\left(|\tau_{\Ncal_n}^\ast-\tau_2^\ast|>r\right)\le 4n^2 \log^4 n \exp\left(-{r^2\log n\over 480\alpha_0\sqrt{\log\log n}}\right).
$$
Taking $r=2\sqrt{480\alpha_0}\log^{1/4}\log n$ yields
\begin{equation}
\label{eq:taudiff}
\PP\left(|\tau_{\Ncal_n}^\ast-\tau_2^\ast|>2\sqrt{480\alpha_0}\log^{1/4}\log n\right)\to 0.
\end{equation}

On the other hand, to bound $\tau_{\Ncal_n}^\ast$, we now appeal to the following lemma. 
\begin{lemma}
\label{lm:singleconcen}
If $X$ and $Y$ are independent uniform random variables, then
\begin{equation}
\label{eq:singleconcen}
\PP\left(\tilde{\tau}(t_X,t_Y)>r|n_{t_X,t_Y}\ge 2 \right)\le \exp\left(-{r^2\over 9}\right),
\end{equation}
where $\tilde{\tau}(t_X,t_Y)$ is defined in (\ref{eq:tildetau}).
\end{lemma}

An application of union bounds and Lemma~\ref{lm:singleconcen} yields
\begin{align*}
\PP\left(\tau_{\Ncal_n}^\ast>r\right)\le & \log^4 n \sup_{(j,k)\in \Ncal_n} \PP\left(\tilde{\tau}(X_{(j)},Y_{(k)})>r\right)\\
\le &\log^4 n\sup_{j,k } \PP\left(\tilde{\tau}(X_{(j)},Y_{(k)})>r\right)\\
\le & \log^4 n\sup_{j,k } \PP\left(\tilde{\tau}(X_{(j)},Y_{(k)})>r|n_{X_{(j)},Y_{(k)}}\ge 2\right)\\
\le & \log^4 n \exp\left(-{r^2\over 9}\right)
\end{align*}
Taking $r=7\sqrt{\log\log n}$ leads to
\begin{equation}
\label{eq:tausn}
\PP\left(\tau_{\Ncal_n}^\ast>7\sqrt{\log\log n}\right)\to 0.
\end{equation}
Combined with (\ref{eq:taudiff}) and (\ref{eq:tausn}), we obtain
$$
\PP\left(\tau_2^\ast>8\sqrt{\log\log n}\right)\to 0.
$$

Finally, we consider $\tau_3^\ast$, which turns out to be the most complex. We first group $(j,k)$ according to its size by defining the following collection of coordinates:
$$
\Tcal_n(M_i,\eta_i)=\{(j,k):M_i/\eta_i\le n(1-j/n)(1-k/n)\le M_i\}, 
$$
where $1\le i\le Q_n$.
Here,  $Q_n$ is the smallest integer such that $M_{Q_n+1}<\alpha^2_0\log\log n$, and $\eta_i$ and $M_i$ be two positive sequences such that 
$$
M_1=\log^2 n,\quad \eta_i=1+{\sqrt{\log\log n}\over \sqrt{M_i}}\quad{\rm and}\quad M_{i+1}=M_i/\eta_i.
$$
It is not hard to see that
$$
\tau_3^\ast=\max_{1\le i\le Q_n} \max_{(j,k)\in \Tcal_n(M_i,\eta_i)}\tilde{\tau}(X_{(j)},Y_{(k)}).
$$
We employ a strategy similar to the previous case to bound $\max_{(j,k)\in \Tcal_n(M_i,\eta_i)}\tilde{\tau}(X_{(j)},Y_{(k)})$ for each $i$. To this end, we define the following approximation set to $\Tcal_n(M_i,\eta_i)$:
\begin{align*}
&\Scal_n(M_i,\eta_i)=\{(j,k):(j,k){\rm\ or \ }(k,j) {\rm\ is\ of\ form}\\
&\left(\lfloor n-n/2\eta_i^u\rfloor,\lfloor n-2M_i\eta_i^{u+1}\rfloor\right),u=0,\ldots, \left\lfloor{\log n\over 2\log \eta_i}\right\rfloor\}.
\end{align*}

We first bound to the difference between max on $\Scal_n(M_i,\eta_i)$ and on $\Tcal_n(M_i,\eta_i)$. Similar to before, we consider a class of maps $\pi_{i}:\Tcal_n(M_i,\eta_i)\to\Scal_n(M_i,\eta_i)$ such that
$$
\pi_{i}(j,k)=\begin{cases}\left(\lfloor n-n/2\eta_i^{h(j)}\rfloor,\lfloor n-2M_i\eta_i^{h(j)+1}\rfloor\right) &{\rm if}\ j>k \\ \left(\lfloor n-2M_i\eta_i^{h(k)+1}\rfloor,\lfloor n-n/2\eta_i^{h(k)}\rfloor\right) &{\rm if}\ j\le k \end{cases}
$$
where $h(t)$ maps $t$ to the largest $u$ such that $\lfloor n-n/2\eta_i^u\rfloor<t$.
For any $(j,k)\in \Scal_n(M_i,\eta_i)$, its pre-image $\pi^{-1}_i(j,k)$ is the collection of all pairs $(j',k')$ in $\Tcal_n(M_i,\eta_i)$ which satisfy $\pi_i(j',k')=(j,k)$. Moreover, we define a conjugate pair of $(j,k)\in \Scal_n(M_i,\eta_i)$ as
$$
c(j,k)=\begin{cases}\left(\lfloor n-n/2\eta_i^{u+1}\rfloor,\lfloor n-2M_i\eta_i^{u}\rfloor\right) &{\rm if}\ j>k \\ \left(\lfloor n-2M_i\eta_i^{u}\rfloor,\lfloor n-n/2\eta_i^{u+1}\rfloor\right) &{\rm if}\ j\le k \end{cases}
$$
if $u$ is an integer such that $(j,k){\rm\ or \ }(k,j)=\left(\lfloor n-n/2\eta_i^u\rfloor, \lfloor n-2M_i\eta_i^{u+1}\rfloor\right)$. For simplicity, denote by $c(j,k)_x$ and $c(j,k)_y$ the two indices of $c(j,k)$.
Figure \ref{fg:map} gives a specific example to illustrate the idea behinf $\pi_i$, $\pi^{-1}_i(j,k)$ and conjugate pair $c(j,k)$.

As shown in Figure~\ref{fg:map}, for any $(j,k)\in \Scal_n(M_i,\eta_i)$, 
$$
\bigcup_{(j',k')\in \pi^{-1}_i(j,k)}\Kcal(X_{(j')},Y_{(k')})\subset \Kcal(X_{(j)},Y_{(k)}),
$$
where, recall that $\Kcal(t_X,t_Y)=\left\{i\in \II:X_i\ge t_X,Y_i\ge t_Y\right\}$. This suggests that 
\begin{equation}
\label{eq:numineq}
\begin{split}
&\left|\left\{\tilde{\tau}(X_{(j')},Y_{(k')}):(j',k')\in \pi^{-1}_i(j,k)\right\}\right|\\
\le &\left(n_{X_{(j)},Y_{(k)}}-n_{X_{(c(j,k)_x)},Y_{(c(j,k)_y)}}\right)^2.
\end{split}
\end{equation}
Here $|\cdot|$ represents the cardinality of a set. By a similar argument as that for (\ref{eq:numineq1}), we have 
\begin{equation}
\label{eq:numminus}
n_{X_{(j)},Y_{(k)}}\le n_{X_{(c(j,k)_x)},Y_{(c(j,k)_y)}}\left(1+{20\alpha_0\sqrt{\log\log n}\over \sqrt{M_i}}\right),
\end{equation}
for any $(j,k)\in \Scal_n(M_i,\eta_i)$.

Equations (\ref{eq:numineq}) and (\ref{eq:numminus}) together imply that
\begin{align*}
&\left|\left\{\tilde{\tau}(X_{(j')},Y_{(k')}):(j',k')\in \pi^{-1}_i(j,k)\right\}\right|\\
\le &(20\alpha_0n_{X_{(j)},Y_{(k)}}\sqrt{\log\log n}/\sqrt{M_i})^2\\
\le  &\log^3 n.
\end{align*}
This means that the number of distinct values among $\{\tilde{\tau}(X_{(j')},Y_{(k')}): (j',k')\in \pi^{-1}_i(j,k)\}$ is not very large. We can then apply union bound, (\ref{eq:numminus}) and Lemma~\ref{lm:diffconcen} to get
\begin{align*}
&\PP\left(\left|\max_{\substack{(j,k)\in\\  \Tcal_n(M_i,\eta_i)}} \tilde{\tau}(X_{(j)},Y_{(k)})-\max_{\substack{(j,k)\in\\  \Scal_n(M_i,\eta_i)}} \tilde{\tau}(X_{(j)},Y_{(k)})\right|>r\right)\\
\le & \sum_{\substack{(j,k)\in\\ \Scal_n(M_i,\eta_i)}}\PP\left(\max_{\substack{(j',k')\in\\ \pi^{-1}_i(j,k) }}\left|\tilde{\tau}(X_{(j)},Y_{(k)})-\tilde{\tau}(X_{(j')},Y_{(k')})\right|>r\right)\\
\le & \sum_{(j,k)\in \Scal_n(M_i,\eta_i)} 4\log^3 n \exp\left(-{r^2\over 24(\eta_i-1)}\right)\\
\le & 4\log^5 n \exp\left(-{r^2\over 24(\eta_i-1)}\right).
\end{align*}

Recall that 
$$
Q_n\le T_1\le \log^2 n\quad {\rm and}\quad \eta_i\le 1+{1\over \alpha_0}.
$$
Therefore,
\begin{equation}
\label{eq:alldiff}
\begin{split}
&\PP\left(\left|\max_{\substack{(j,k)\in \\ \cup_{i=1}^{Q_n}\Tcal_n(M_i,\eta_i)}} \tilde{\tau}(X_{(j)},Y_{(k)})-\max_{\substack{(j,k)\in \\ \cup_{i=1}^{Q_n}\Scal_n(M_i,\eta_i)}} \tilde{\tau}(X_{(j)},Y_{(k)})\right|>r\right)\\
\le &\sum_{i=1}^{Q_n}\PP\left(\left|\max_{\substack{(j,k)\in \\ \Tcal_n(M_i,\eta_i)}} \tilde{\tau}(X_{(j)},Y_{(k)})-\max_{\substack{(j,k)\in\\  \Scal_n(M_i,\eta_i)}} \tilde{\tau}(X_{(j)},Y_{(k)})\right|>r\right)\\
\le & 4\log^7 n \exp\left(-{\alpha_0 r^2\over 24}\right)
\end{split}
\end{equation}
It is clear that
\begin{align*}
&\left\{(j,k):j,k\ge \lfloor n/2\rfloor,{\alpha^2_0\log\log n\over n}\le L\left({j\over n},{k\over n}\right)\le {\log^2n\over n} \right\}\\
&\subset \bigcup_{i=1}^{Q_n} \Tcal_n(M_i,\eta_i)
\end{align*}
Taking $r=16\sqrt{\log\log n}/\alpha_0$ in (\ref{eq:alldiff}) yields
\begin{equation}
\label{eq:diff3}
\PP\left(\left|\tau_3^\ast-\max_{\substack{(j,k)\in\\  \cup_{i=1}^{Q_n}\Scal_n(M_i,\eta_i)}} \tilde{\tau}(X_{(j)},Y_{(k)})\right|>{16\sqrt{\log\log n}\over\alpha_0}\right)\to 0
\end{equation}

\begin{figure}[htbp]
\begin{center}
\includegraphics[width=7cm]{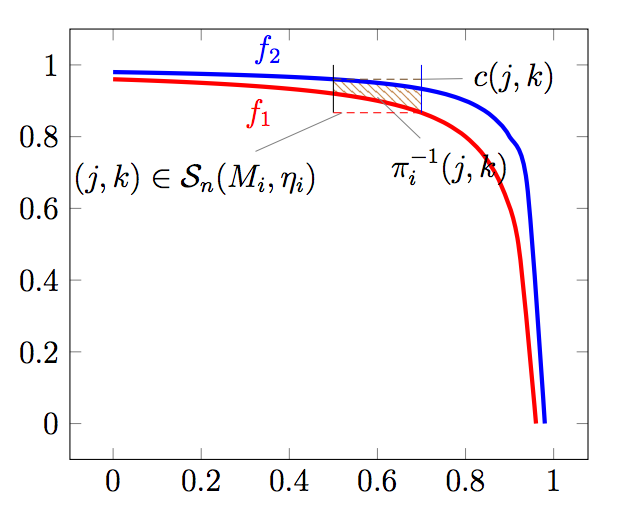}
\end{center}
\caption{$f_1:n(1-j/n)(1-k/n)= M_i$ and $f_2:n(1-j/n)(1-k/n)= M_i/\eta_i$, The shadded area is  $\pi^{-1}_i(j,k)$ for some $(j,k)\in \Scal_n(M_i,\eta_i)$.}
\label{fg:map}
\end{figure}

An application of union bound and Lemma~\ref{lm:singleconcen} then yields
\begin{align*}
&\PP\left(\max_{(j,k)\in  \cup_{i=1}^{Q_n}\Scal_n(M_i,\eta_i)} \tilde{\tau}(X_{(j)},Y_{(k)})>r\right)\\
\le& \sum_{i=1}^{Q_n}\PP\left(\max_{(j,k)\in \Scal_n(M_i,\eta_i)} \tilde{\tau}(X_{(j)},Y_{(k)})>r\right)\\
\le& \log^4 n \sup_{(j,k)\in  \cup_{i=1}^{Q_n}\Scal_n(M_i,\eta_i)}\PP\left(\tilde{\tau}(X_{(j)},Y_{(k)})>r\right)\\
\le& \log^4 n \exp\left(-{r^2\over 9}\right)
\end{align*}
Taking $r=7\sqrt{\log\log n}$ yields
\begin{equation}
\label{eq:union3}
\PP\left(\max_{(j,k)\in  \cup_{i=1}^{Q_n}\Scal_n(M_i,\eta_i)} \tilde{\tau}(X_{(j)},Y_{(k)})>7\sqrt{\log\log n}\right)\to 0
\end{equation}
Together with (\ref{eq:diff3}), it implies that
$$
\PP\left(\tau_3^\ast>7\sqrt{\log\log n}+{16\sqrt{\log\log n}\over \alpha_0}\right)\to 0.
$$
and we can complete proof. The statement about $q_\alpha$ then follows from the bounds we derived for $\tau_1^\ast$, $\tau_2^\ast$ and $\tau_3^\ast$.

\paragraph{Type II eror.} To prove the first statement, it now suffices to show that under $H_0$, if
$$
V(t_X,t_Y)\cdot T^2(t_X,t_Y) \gg {\log\log n\over n},
$$
then $\tau^\ast\gg \sqrt{\log\log n}$.

Recall that $\Kcal=\{i:X_i>t_X,\ Y_i>t_Y\}$, and
$$
\EE(\sign(X_i-X_j)\sign(Y_i-Y_j)|i,j\in \Kcal)=T(t_X,t_Y).
$$
This suggests that
$$
\EE\left(\tau(t_X,t_Y)\cdot W(n_{t_X,t_Y}) \middle|n_{t_X,t_Y}=k\right)\ge {3\over 4}\sqrt{k}T(t_X,t_Y).
$$
Here, $W(x):=\sqrt{9x(x-1)\over 2(2x+5)}$.
Clearly, $n_{t_X,t_Y}$ follows binomial distribution $n_{t_X,t_Y}\sim Bin(n,\theta)$, where $\theta=1+F(t_X,t_Y)-F_X(t_X)-F_Y(t_Y)$. It is easy to derive from Chernoff's bounds that
$$
\PP\left({n\theta\over 2}<n_{t_X,t_Y}<{3n\theta\over 2}\right)\to 1.
$$
This implies that
$$
\PP\left(\tau(t_X,t_Y)\cdot W(n_{t_X,t_Y})>\sqrt{9 \over 32}\sqrt{n\theta}T(t_X,t_Y)\right)\to 1.
$$
and hence
$$
\PP\left(\tau(t_X,t_Y)\cdot W(n_{t_X,t_Y})\gg \sqrt{\log\log n}\right)\to 1.
$$
It follows that
$$
\tau^\ast\ge \tau(t_X,t_Y)\cdot \sqrt{9n_{t_X,t_Y}(n_{t_X,t_Y}-1)\over 2(2n_{t_X,t_Y}+5)}\gg\sqrt{\log\log n}),
$$
with probability tending to one.

\paragraph{Lower bound.} To show that we can not detect a signal under the condition (\ref{eq:lowercond}), we consider a special case where under the null, $F_0(x,y)=xy\II(0\le x\le 1,0\le y\le 1)$; and under the alternative, the joint distribution comes from variants of Farlie-Gumbel-Morgenstern family so that its density can be given by 
\begin{align*}
&f_{(U_j,\gamma_j)}\\
=&{dF_{(U,\gamma)}(x,y)\over dxdy}\\
=&\begin{cases}1+\gamma\left(1-2\left(x-U\over 1-U \right)\right)\left(1-2\left(y-U\over 1-U \right)\right) &  U\le x,y\le 1 \\1 & {\rm otherwise} \end{cases},
\end{align*}
where $U$ and $\gamma$ is some constant between $0$ and $1$.

Let $(U_j,\gamma_j)$ be a sequence of pairs such that $U_j=1-1/2^j$ and $\gamma_j(1-U_j)=2\sqrt{\log\log n}/\sqrt{n}$ and let $M_n=\lfloor\log n/4 \rfloor$.
It is not hard to verify that $F_{(U_j,\gamma_j)}$s satisfy (\ref{eq:lowercond}) with $c=16/81$ by noting
$$
V(U_j,U_j)\cdot T^2(U_j,U_j)={4\gamma_j^2\over 81}(1-U_j)^2.
$$
Denote by $\PP_0$ the joint distribution of $(X_i,Y_i)_{i=1}^{n}$ with distribution $F_0(x,y)$ and, for $j=1,\ldots,M_n$, $\PP_j$ joint distribution of $(X_i,Y_i)_{i=1}^{n}$ with density distribution of $F_{(U_j,\tau_j)}(x,y)$.
Then, the likelihood ratio between $\PP_j$ and $\PP_0$ is
$$
L_j:={d\PP_j \over d\PP_0}=\prod_{i=1}^n f_{(U_j,\gamma_j)}(X_i,Y_i).
$$
Elementary calculations lead to
\begin{equation}
\label{eq:crossproduct}
\begin{split}
&\EE_{0}(f_{(U_j,\gamma_j)}(X_i,Y_i)f_{(U_k,\gamma_k)}(X_i,Y_i))\\
=&1+{\gamma_j(1-U_j)\gamma_k(1-U_k)\over 9*2^{3|j-k|}}\\
=&1+{4\log\log n\over 9n*2^{3|j-k|}},
\end{split}
\end{equation}
where $\EE_0$ stands s for expectation taken with respect to $\PP_0$.
By definition and (\ref{eq:crossproduct}), we can ensure
\begin{align*}
\EE_{0}(L_jL_k)-1&=\prod_{i=1}^n \EE_{0}(f_{(U_j,\gamma_j)}(X_i,Y_i)f_{(U_k,\gamma_k)}(X_i,Y_i))-1\\
&=\left(1+{4\log\log n\over 9n\times 2^{3|j-k|}}\right)^n-1\\
&\le {1\over 2^{3|j-k|}}\left(\left(1+{4\log\log n\over 9n}\right)^n-1\right)\\
&\le {1\over 2^{3|j-k|}}\left(\exp\left({4\log\log n\over 9}\right)-1\right).
\end{align*}
This immediately suggests that
$$
\sum_{j,k=1}^{M_n}\left(\EE_{0}(L_jL_k)-1\right)\le 2M_n \exp\left({4\log\log n\over 9}\right).
$$
Then, by Jensen's inequality, we have
\begin{equation}
\label{eq:converge}
\begin{split}
&\EE_{0}\left(\left|{1\over M_n}\sum_{j=1}^{M_n} L_j-1\right|\right)^2\\
\le& \EE_{0}\left(\left|{1\over M_n}\sum_{j=1}^{M_n} L_j-1\right|^2\right)\\
\le& 2M_n^{-1} \exp\left({4\log\log n\over 9}\right)\to 0
\end{split}
\end{equation}
Let $\phi$ be any test that depends on $\{X_i,Y_i\}_{i=1}^n$. Then, by (\ref{eq:converge}), 
\begin{align*}
&\max_{j=1,\ldots,M_n}\PP_j(\phi=0)+\PP_0(\phi=1)\\
=&1-\min_{j=1,\ldots,M_n}\EE_{j}\phi+\EE_0(\phi)\\
\ge& 1-{1\over M_n}\sum_{j=1}^{M_n}\EE_{j}\phi+\EE_{0}\phi\\
=&1-\EE_{0}\left({1\over M_n}\sum_{j=1}^{M_n}L_j-1\right)\phi\\
\ge& 1-\EE_{0}\left|{1\over M_n}\sum_{j=1}^{M_n}L_j-1\right|\\
\to& 1
\end{align*}
Here $\EE_j$ stands s for expectation taken with respect to $\PP_j$; we complete the proof.

\appendix[Proof of Technical Lemmas]
\begin{proof}[Proof of Lemma \ref{lm:diffconcen}]
First consider a simple case where $(t_{X}-t'_{X})(t'_{Y}-t'_{Y})=0$ and assume $t_{X}=t'_{X}$ and $t_{Y}\le t'_{Y}$ without loss of generality. Let $S=\{(t,s):t_{X}\le t\le 1,t_{Y}\le s\le 1\}$ and $S'=\{(t,s):t_{X}\le t\le 1,t'_{Y}\le s\le 1\}$. 
We randomly choose points $(\tilde{X}_i,\tilde{Y_i})_{i=1}^{k'}$ in $S'$ and points $(\tilde{X}_i,\tilde{Y_i})_{i=k'+1}^{k}$ in $S\setminus S'$, where $k\le (1+\epsilon)k'$.
Condition on $\{n_{t_X,t_Y}=k,n_{t'_X,t'_Y}=k'\}$, $\tilde{\tau}(t_X,t_Y)-\tilde{\tau}(t'_X,t'_Y)$ has the same distribution as the following statistic:
\begin{align*}
& f((\tilde{X}_i,\tilde{Y}_i)_{i=1}^{k})\\
=&{3\sqrt{2}\over \sqrt{k(k-1)(2k+5)}}\sum_{1\le i<j\le k}\sign(\tilde{X}_i-\tilde{X}_j)\sign(\tilde{Y}_i-\tilde{Y}_j)\\
-&{3\sqrt{2}\over \sqrt{k'(k'-1)(2k'+5)}}\sum_{1\le i<j\le k'}\sign(\tilde{X}_i-\tilde{X}_j)\sign(\tilde{Y}_i-\tilde{Y}_j)
\end{align*}

Next, we show that $f$ has bounded difference with respect to $(\tilde{X}_i,\tilde{Y_i})$. Write
\begin{align*}
\Delta_j:=&\sup_{(\tilde{X}_i,\tilde{Y}_i)_{i=1}^{k},(\tilde{X}'_j,\tilde{Y}'_j)}|f((\tilde{X}_i,\tilde{Y}_i)_{1\le i\le k})\\
&-f((\tilde{X}_i,\tilde{Y}_i)_{1\le i\le k,i\ne j},(\tilde{X}'_j,\tilde{Y}'_j))|
\end{align*}
When $1\le j\le k'$,
$$
\Delta_j\le 3\left({1\over \sqrt{k'}}-{k'\over k\sqrt{k}}+{k-k'\over k\sqrt{k}}\right)\le {6(k-k')\over k\sqrt{k'}}
$$
and, when $k'<j\le k$,
$$
\Delta_j\le {3k'\over k\sqrt{k}}.
$$
Because $k\le (1+\epsilon)k'$, we have
$$
v={1\over 4}\sum_j\Delta_j^2\le {9(k-k')^2\over k^2}+{9(k-k')\over 4k}\le 9\epsilon^2+{9\over 4}\epsilon
$$
Applying McDiarmid inequality \citep[see, e.g.,][]{boucheron2013} to $f$,
\begin{align*}
&\PP\left(\tilde{\tau}(t_X,t_Y)-\tilde{\tau}(t'_X,t'_Y)>r|n_{t_X,t_Y}=k,n_{t'_X,t'_Y}=k' \right)\\
\le &\exp\left(-{r^2\over 18\epsilon^2+9\epsilon/2}\right)
\end{align*}
where we used the fact that $\EE(f((\tilde{X}_i,\tilde{Y}_i)_{i=1}^{k}))=0$. By symmetry,
\begin{align*}
&\PP\left(|\tilde{\tau}(t_X,t_Y)-\tilde{\tau}(t'_X,t'_Y)|>r| B \right)\\
\le& 2\PP\left(\tilde{\tau}(t_X,t_Y)-\tilde{\tau}(t'_X,t'_Y)>r| B \right)\\
\le & 2\exp\left(-{r^2\over 18\epsilon^2+9\epsilon/2}\right).
\end{align*}

Next, we consider the case when $t_{X}>t'_{X}$ and $t'_{Y}>t'_{Y}$ and all other remaining cases can be treated in an identical fashion. Applying the result for $(t_{X}-t'_{X})(t'_{Y}-t'_{Y})=0$, we can derive that
\begin{align*}
&\PP\left(|\tilde{\tau}(t_X,t_Y)-\tilde{\tau}(t'_X,t'_Y)|>r| B \right)\\
\le& \PP\left(|\tilde{\tau}(t_X,t_Y)-\tilde{\tau}(t_{X},t'_{Y})|>r/2| B \right)\\
&+\PP\left(|\tilde{\tau}(t_{X},t'_{Y})-\tilde{\tau}(t'_X,t'_Y)|>r/2| B \right)\\
\le & 4\exp\left(-{r^2\over 72\epsilon^2+18\epsilon}\right).
\end{align*}
This completes the proof.
\end{proof}
\vskip 15pt

\begin{proof}[Proof of Lemma \ref{lm:singleconcen}]
Note that
\begin{align*}
&\PP\left(\tilde{\tau}(t_X,t_Y)>r|n_{t_X,t_Y}\ge 2 \right)\\
=&{\sum_{k=2}^{n} \PP\left(\tilde{\tau}(t_X,t_Y)>r|n_{t_X,t_Y}=k \right)\PP(n_{t_X,t_Y}=k) \over \PP(n_{t_X,t_Y}\ge 2)}\\
\le& \sup_{k\ge 2} \PP\left(\tilde{\tau}(t_X,t_Y)>r|n_{t_X,t_Y}=k \right)
\end{align*}
Thus it is sufficient to set an upper bound to $\PP\left(\tilde{\tau}(t_X,t_Y)>r|n_{t_X,t_Y}=k \right)$.
Condition on $n_{t_X,t_Y}=k$, $\tilde{\tau}(t_X,t_Y)$ has the same distribution with
$$
{3\sqrt{2}\over \sqrt{k(k-1)(2k+5)}}\sum_{1\le i<j\le k}\sign(\tilde{X}_i-\tilde{X}_j)\sign(\tilde{Y}_i-\tilde{Y}_j)
$$
where $(\tilde{X}_i,\tilde{Y}_i)_{i=1}^k$ comes from distribution given $X>t_X$ and $Y>t_Y$. Using the concentration inequality for U-statistics from \cite{hoeffding1963}, we get
\begin{align*}
&\PP\left(\tilde{\tau}(t_X,t_Y)-\EE(\tilde{\tau}(t_X,t_Y))|n_{t_X,t_Y}=k)>r|n_{t_X,t_Y}=k \right)\\
&\le e^{-r^2/9}.
\end{align*}
The proof is then completed noting that $\tilde{X}_i$ and $\tilde{Y}_i$ remain independent of each other.
\end{proof}


\ifCLASSOPTIONcaptionsoff
  \newpage
\fi

\bibliographystyle{IEEEtran}  
\bibliography{RankCorrelation} 

\end{document}